\documentclass[amssymb, nofootinbib, superscriptaddress]{revtex4}
\usepackage{hyperref}
\usepackage{amsfonts,amssymb,amsmath,exscale,bbm}
\usepackage{footnpag} 
\usepackage[all,knot]{xy}  
\usepackage{color}
\usepackage{graphicx}
\usepackage{epsfig}
\usepackage{subfig}

\usepackage{amsmath}
\usepackage{amsfonts}
\usepackage{graphicx}
\usepackage{psfrag}
\usepackage{array}
\usepackage{bbm}
\usepackage{hyperref}
\usepackage{amsthm}
\theoremstyle{definition}

\newtheorem{theorem}{Theorem}

\newtheorem{proposition}{Proposition}

\newcommand{\cA}{{\mathcal A}}
\newcommand{\cB}{{\mathcal B}}
\newcommand{\cC}{{\mathcal C}}
\newcommand{\cD}{{\mathcal D}}

\newcommand{\cF}{{\mathcal F}}
\newcommand{\cG}{{\mathcal G}}

\newcommand{\cK}{{\mathcal K}}

\newcommand{\cN}{{\mathcal N}}

\newcommand{\cP}{{\mathcal P}}

\newcommand{\cS}{{\mathcal S}}
\newcommand{\cT}{{\mathcal T}}
\newcommand{\cV}{{\mathcal V}}

\newcommand{\cZ}{{\mathcal Z}}

\def\inv{{\mbox{\tiny -1}}}

\newcommand{\R}{\mathbb{R}}
\newcommand{\C}{\mathbb{C}}

\newcommand\beq{\begin{equation}}
\newcommand\eeq{\end{equation}}
\newcommand{\be}{\begin{equation}}
\newcommand{\ee}{\end{equation}}
\newcommand{\bes}{\begin{eqnarray}}
\newcommand{\ees}{\end{eqnarray}}

\def\cc{{\cal C}}

\def\ot{{\,\otimes \,}}

\def\act{\rhd}

\def\vphi{{\varphi}}
\def\tvphi{{\tilde \varphi}}

\newcommand{\one}{\mbox{$1 \hspace{-1.0mm}  {\bf l}$}}

  \def\cc{{\cal C}}    
      \def\nn{{\nonumber}}

\def\hpsi{{\widehat \psi}}

\def\act{{\, \triangleright\, }}

\newcommand{\su}{\mathfrak{su}}

\newcommand{\SU}{\mathrm{SU}}

\newcommand{\U}{\mathrm{U}}

\newcommand{\SO}{\mathrm{SO}}

\def\extd{\mathrm {d}}
\newcommand{\e}{\epsilon}

\newcommand\acts\triangleright

\newcounter{letter} \newcounter{numeral} \newcounter{Numeral}

\newcommand\Tr{\mathrm{Tr}}
\def\vphi{\varphi}
\def\vphihat{\widehat{\varphi}}
\def\tvphihat{\widehat{\tilde{\varphi}}}
\def\tpsi{\tilde{\psi}}

\def\e{\mbox{e}}
\def\extd{\mathrm {d}}

\newcommand\maps{\colon}

\newtheorem{theo}{Theorem}
\newtheorem{lemma}[theo]{Lemma}

\begin{document}

%

%

\title{\Large \bf Bounding bubbles: {\large the vertex representation of 3d Group Field Theory and the suppression of pseudo-manifolds}}

\author{Sylvain Carrozza}\email{sylvain.carrozza@aei.mpg.de}
\author{Daniele Oriti}\email{daniele.oriti@aei.mpg.de}\affiliation{Max Planck Institute for Gravitational
Physics, Albert Einstein Institute, Am M\"uhlenberg 1,
14476 Golm, Germany, EU}
\affiliation{Laboratoire de Physique Th\'{e}orique, CNRS UMR 8627,
Universit\'{e} Paris XI, F-91405 Orsay Cedex, France, EU}



\begin{abstract}
\noindent 

Based on recent work on simplicial diffeomorphisms in colored group field theories, we develop a representation of the colored Boulatov model, in which the GFT fields depend on variables associated to vertices of the associated simplicial complex, as opposed to edges. On top of simplifying the action of diffeomorphisms, the main advantage of this representation is that the GFT Feynman graphs have a different stranded structure, which allows a direct identification of subgraphs associated to bubbles, and their evaluation is simplified drastically. As a first important application of this formulation, we derive new scaling bounds for the regularized amplitudes, organized in terms of the genera of the bubbles, and show how the pseudo-manifolds configurations appearing in the perturbative expansion are suppressed as compared to manifolds. Moreover, these bounds are proved to be optimal.
\end{abstract}

\maketitle
\bigskip \bigskip

\section{Introduction and background}
Group field theories (GFTs)\cite{gftreview, iogft2, Rivasseau:2011xg} are a d-dimensional generalization of matrix models for 2d gravity \cite{mm} in the form of field theories over group manifolds. Moreover, they enter conspicuously in  the definition of the dynamics of loop quantum gravity \cite{carlo, thomas}, and share many conceptual and mathematical ingredients with simplicial quantum gravity approaches, like quantum Regge calculus \cite{qRC} and dynamical triangulations \cite{DT}, bringing them in the position to profit from the achievements and insights of all these other approaches \cite{libro}. 

As matrix models, they are characterized by a combinatorial pattern of identifications of field arguments in the interaction, such that the perturbative expansion of the theory generates a sum over d-dimensional simplicial complexes. For quantum gravity models, these simplicial complexes represent discrete spacetimes and the perturbative sum is expected to provide a definition of a covariant dynamics of quantum gravity in d dimensions, i.e. a sum over geometries.
Matrix models \cite{mm} succeed in doing so, as said, in the simple case of 2d Euclidean quantum gravity; already in this simple case, it has been a rather non-trivial task, which moreover led to development of very powerful tools and a pletora of further applications. This success rests on four main (sets of) achievements: 1) the Feynman amplitudes that the models associate to the 2d simplicial complexes generated in perturbative expansion can be directly related to simplicial gravity path integrals (coupled to matter) weighted by the Regge action for equilateral triangulations; therefore a clear link with gravity and geometry is ensured already at the discrete level, and this guides both the development and the interpretation of the theory; 2) the perturbative sum over simplicial complexes can be controlled in the sense that models can be written in which only simplicial manifolds are generated (it is enough to ensure orientability) and, most importantly, it can be organized as a topological expansion; thanks to this, one can identify a regime (large dimension $N$ of the matrices) in which simple topologies dominate; 3) a continuum (thermodynamic) limit of the models can be defined, both for trivial topologies only and admitting the contribution of all topologies, for appropriate critical values of the parameters of the models; 4) in this continuum limit, one is able to match the quantum dynamics of the matrix model (transition amplitudes and their Schwinger-Dyson equations, critical exponents, thermodynamical quantities, etc) with quantum geometrodynamics (Wheeler-DeWitt equation, continuum path integral for given topology, etc) and, then, semi-classical gravity coupled to matter in two dimensions.

The first attempt at a generalization led to tensor models \cite{tensor}, based on the same basic idea, but with matrices replaced by tensors, with index pairing in the interaction such that their perturbative expansion would generate d-dimensional simplicial complexes. However, while tensor models are still being developed with interesting applications \cite{sasakura}, they could not reproduce any of the above crucial steps towards success as quantum gravity models, as matrix models did. No direct link with discrete quantum gravity and, most problematic, no control over the perturbative expansion, meant that no continuum limit and no link with continuum gravity could be obtained. In particular, concerning the sum over simplicial complexes generated in perturbation theory, one should notice that: a) all topologies are generated in the expansion and the classification of topologies is an open mathematical problem in 3d, and a known impossibility in higher dimensions; b) alongside manifold configurations, tensor models generate all sorts of more pathological configurations \cite{tensor, DP-P}. The big issue would be to discriminate between all these structures and somehow show that only manifolds of some nice topology dominate. This has not been possible for tensor models. 

Two basic (somewhat complementary) attitudes can be taken in front of this failure. One is to leave aside for the moment the issue of generating the sum over simplicial complexes by some field-theoretic mechanism and instead {\it define} the model as a sum over (equilateral) triangulations of a spherical topology weighted by the Regge action. This leads to the (causal) Dynamical Triangulations approach \cite{DT}. The second identifies the source of the problems in the lack, in tensor models, of enough degrees of freedom to capture the greater complexity of higher-dimensional spacetimes and geometries, and thus goes in the direction of identifying first and then incorporating the correct additional degrees of freedom. This leads to Group Field Theories.

The kind of degrees of freedom to be included and the way to do it is suggested by loop quantum gravity \cite{carlo, thomas}. This is the most advanced canonical quantization of continuum gravity and, starting from a reformulation of gravity as a gauge theory of the Lorentz connection,  has identified the kinematical states of quantum space to be spin networks, i.e. graphs labelled by irreducible representations of the Lorentz group. The dynamics of the same states is given covariantly in terms of spin foam amplitudes \cite{SF}, i.e. functions of the same representations, assigned to 2-cells of cellular complexes representing each a possible discrete history of a spin network state, which should be then summed over to recover the full dynamics, in the spirit of a sum over geometries. Both at the level of quantum states and at the level of their dynamics, then, the basic variables of the theory are either group elements, interpreted as elementary parallel transports of a Lorentz connection, or group representations, interpreted as quantum numbers of geometric observables. These are then the degrees of freedom that are added to tensor models in the group field theory formalism, the basic field being indeed a function on the corresponding group manifold, which could be understood as a second quantization of an elementary spin network wavefunction \cite{gftreview, 3rd}. In spin foam models the cellular complex defining a possible evolution process of a spin network is usually taken to be topologically dual to a simplicial complex (which implies some combinatorial restriction on  both spin network states and cellular complex itself). Remarkably, one can then show that for any spin foam model, i.e. for any choice of dynamical amplitudes, there exists a group field theory which generates it as a Feynman amplitude associated to the simplicial complexes obtained in perturbative expansion. 

Because of this choice of combinatorial structures and because the most studied spin foam models themselves are obtained by quantization of simplicial geometry, one would expect a strict relation between the spin foam amplitudes, and thus the corresponding group field theory, and simplicial gravity path integrals. This relation has been clarified and strengthened by the recent non-commutative metric representation of group field theories \cite{aristidedaniele}, based on the so-called group Fourier transform \cite{PR3,majidfreidel,karim}, a very natural construction on the type of phase space used in loop quantum gravity \cite{ACZ,thomas,fluxes}, Chern-Simons theory \cite{alekseev,hannothomas} and discrete BF theories \cite{eteravalentin,aristidedaniele,danielematti}. 
In this representation, group field theories are written as non-commutative field theories on Lie algebras and the corresponding Feynman amplitudes take explicitly the form of simplicial gravity path integrals, which proves an exact duality between such path integrals and spin foam models. This formulation brings the (quantum) geometry of discrete gravity and of spin foam models to the forefront, and thus is a very convenient starting point for the construction of new models as well as for the physical understanding of existing ones. In fact, it has been crucial \cite{diffeos} in identifying the GFT counterpart of the simplicial gravity transformations that are the  discrete analogue of continuum diffeomorphisms in General Relativity, i.e. translations of vertices of the simplicial complex (in some embedding) that induce transformations of the edge lengths or of the discrete triad (depending on the specific formulation used) associated to the  same simplicial complex, and which leave the gravity action (and solutions of the equations of motion) and the discrete gravity path integral invariant. These are well studied in discrete classical and quantum gravity \cite{qRC,laurentdiffeo,biancadiffeos,biancabenny}, and the main result of \cite{diffeos} has been to identify field transformations of the GFT field that imply these simplicial diffeomorphism transformations at the level of the corresponding Feynman amplitudes, thus of the simplicial gravity path integral. These symmetry transformations also suggest a reformulation of the same GFT model, based on vertex variables on which they act naturally.
In this paper we detail further, in section ~\ref{sec:vertex}, this non-commutative metric formulation in vertex variables, and then use it extensively, from section ~\ref{sec:bubble} onwards, to analyze the dependence of the GFT Feynman amplitudes on the combinatorial structure of the Feynman diagrams, in particular their subgraphs called \lq bubbles\rq and encoding the \lq manifold-ness \rq of their dual simplicial complexes. 

Thus we see that the first kind of achievement of matrix models (and the first failure of simple tensor models) is dealt with successfully by the group field theory formalism. The second type of issues, having to do with the control over the sum over simplicial complexes, is much more thorny, but can now be tackled with the very powerful methods of quantum field theory, in particular those used to study perturbative renormalization, alongside purely combinatorial methods from algebraic topology. In fact, an impressive amount of new results has been obtained recently in this respect, in particular for so-called \lq colored\rq group field theories \cite{RazvanColored}. These go from exact power counting results \cite{lrd,linearizedgft, ValentinMatteo1,ValentinMatteo2,ValentinMatteo3} to perturbative scaling bounds \cite{scaling3d, vincentcolored} and first steps in computing radiative corrections \cite{josephvalentin}, from properties of the combinatorial structures generated \cite{Gurau:2010nd,francesco,MatteoComment} to the important proof that manifolds of spherical topology dominate in the limit of large cut-off (the analogue of the large N limit of matrix models) in any dimension, at least for topological models (not yet for 4d gravity models) \cite{RazvanN,RazvanVincentN,RazvanFullN}. Much remains to be understood, in this respect, and the more we understand even for simpler GFT models the more we will be able to achieve for realistic 4d gravity models. However, it is already clear that GFTs can solve also the second main failure of tensor models, and achieve also the second main set of successes of matrix models. It seems that the incorporation of the key insights of loop quantum gravity and spin foam models was a good move forward. 

In this paper, we contribute further to addressing these topological issues, by taking further advantage of the  improving geometric understanding of the GFT formalism. We show, in section ~\ref{sec:bubble}, that the vertex re-writing of the (colored) Boulatov GFT allows a direct identification of the topology of the 3-cells dual to the vertices of the simplicial complexes, the bubbles, which in turn characterize the \lq manifold-ness\rq  of the complex itself, and a straightforward evaluation of the associated contributions to the GFT Feynman amplitudes. We derive, in sections ~\ref{sec:bubble}, ~\ref{sec:bubble2} and ~\ref{sec:bubble3}, new scaling bounds for the regularized amplitudes, organized in terms of the genera of the bubbles, and show how the pseudo-manifolds configurations appearing in the perturbative expansion are suppressed as compared to manifolds. Moreover, these bounds are proved to be optimal, in section ~\ref{sec:optimal}.  

All of the above is crucial for the general programme of GFT renormalization and thus for the problem of the continuum limit in GFT quantum gravity \cite{gftfluid,Rivasseau:2011xg}. This is another big open issue, of course, and in many ways the decisive one for considering GFTs candidates for a complete theory of quantum gravity and of quantum spacetime. Some results have been obtained recently in this direction, exploring the non-perturbative regime of the theory in particular via mean field methods \cite{eterawinston, noi,emergent,danielelorenzo, ejd}, but it is clear that a more exact evaluation of the GFT partition function, at least for some models, would be desirable and that this will need a more detailed understanding of the combinatorial properties of its perturbative expansion. We hope that the results we present in this paper will also be of help in this respect.

%
\section{Model and translation symmetry}
%
\subsection{Definition of the model}

We consider a (slightly modified) version of the colored bosonic Boulatov model defined in \cite{vincentcolored}. This is a field theory of four complex scalar fields $\{\tvphi_\ell\ , \ell=1,..,4\}$ over three copies of $\SO(3)$, which respect the following gauge invariance:
\beq \label{gauge}
\forall h \in \SO(3),  \qquad \tvphi_\ell(hg_1, hg_2, hg_3)  \, = \, \tvphi_\ell(g_1, g_2, g_3) .
\eeq
They are interpreted as quantum triangles, the $\SO(3)$ variables being interpreted as parallel transports of an $\SO(3)$ connection from the center of the triangle to the (center of the) edges (see figure (\ref{triangle_edge})). We consider an action where the interaction encodes the gluing of these triangles following the pattern of an oriented tetrahedron (see figure (\ref{interaction_edge}))\footnote{The orientation chosen is the only difference between our definition of the model and the one in \cite{vincentcolored}. This choice does not modify any of the results. This interaction term was already considered in \cite{lrd}, but in a non-colored model.}, and the kinetic term is trivial, i.e. it contains only delta functions on the group manifold:

\bes \label{action}
S[\tvphi]&=& S_{kin} [\tvphi] + S_{int}[\tvphi] ,\\
S_{kin} [\tvphi] &=& \int [\extd g_i]^3 \sum_{\ell=1}^4   \, \tvphi_\ell(g_1, g_2, g_3) \overline{\tvphi_{\ell}}(g_1, g_2. g_3) ,\\
S_{int}[\tvphi] &=& \lambda \int [\extd g_{i} ]^6 \, \tvphi_1(g_1, g_2, g_3) \tvphi_2(g_3, g_5, g_4) \tvphi_3(g_5, g_2, g_6) \tvphi_4(g_4, g_6, g_1)
\nn \\
& & +\overline{\lambda} \int [\extd g_{i} ]^6 \, \overline{\tvphi_1}(g_1, g_2, g_3) \overline{\tvphi_2}(g_3, g_5, g_4) \overline{\tvphi_3}(g_5, g_2, g_6) \overline{\tvphi_4}(g_4, g_6, g_1). 
\ees
The ordering of the variables in the fields defines the orientation of the triangles, which we will use shortly to define a symmetry transformation for the fields, so that it is relevant up to even permutations only. However, one could repeat the whole construction with a different choice. The orientability of the Feynman diagrams (simplicial complexes) of the model is proven using the coloring of the same \cite{vincentcolored,francesco} and does not make use of this ordering. Here we oriented the four triangles of the tetrahedral interaction inward (in a right-handed fashion).

\begin{figure}[h]
  \centering
  \subfloat[Field]{\label{triangle_edge}\includegraphics[scale=0.5]{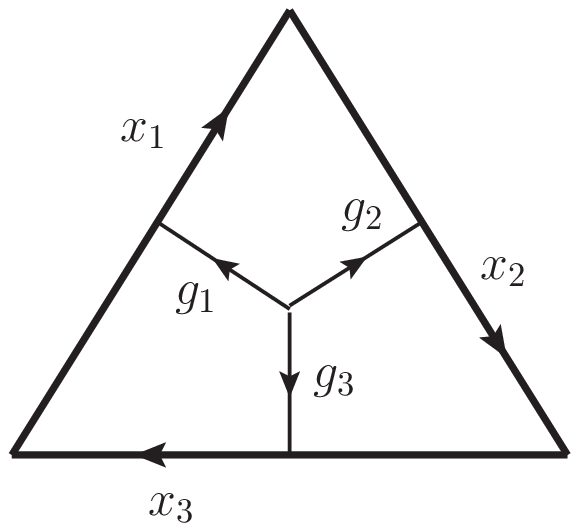}}                
  \subfloat[Interaction vertex (clockwise)]{\label{interaction_edge}\includegraphics[scale=0.5]{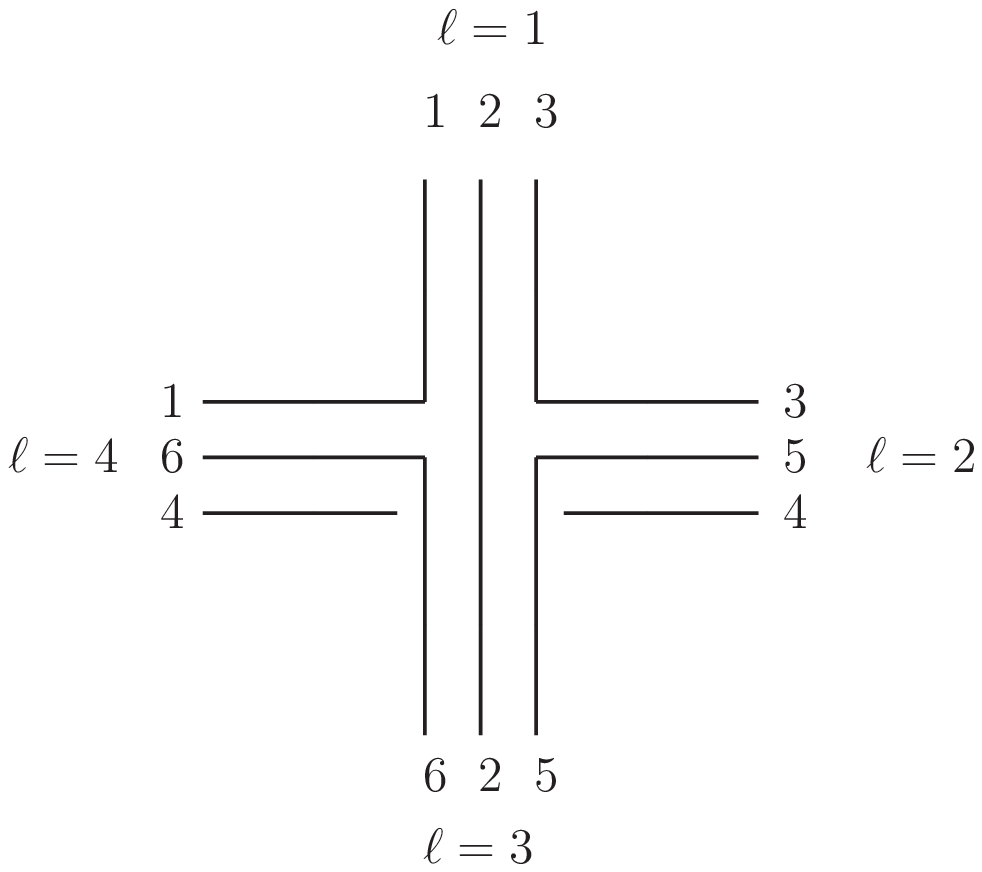}}
  \subfloat[Geometrical interpretation]{\label{tetrahedron_edge}\includegraphics[scale=0.4]{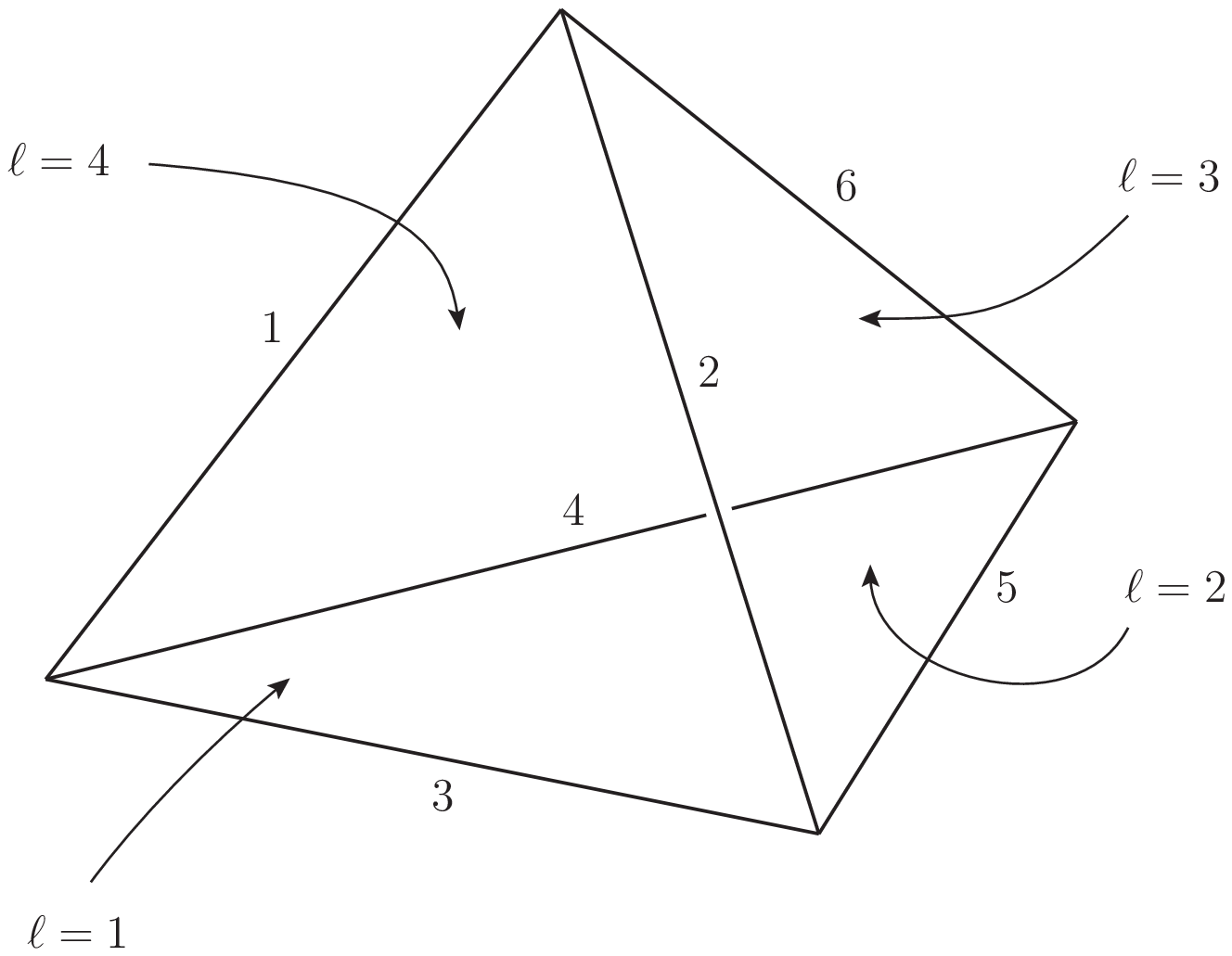}}
  \caption{Graphical representation of a field, and the interaction vertex in usual edge variables.}
  \label{edge_rep}
\end{figure}

Alternatively, we can work with (non-commuting) Lie algebra variables $x \in \su(2) \sim \R^3$, by Fourier transforming the fields as:
\beq \label{fourier}
\tvphihat_\ell(x_1, x_2, x_3) :=\int  [\extd g_i]^3\, \tvphi_\ell(g_1, g_2, g_3)  \,\e_{g_1}(x_1) \e_{g_2}(x_2) \e_{g_3}(x_3) ,
\eeq
where $\e_g \maps \su(2) \!\sim\! \R^3 \to \U(1)$ are non-commutative plane-waves \cite{PR3,majidfreidel,karim}, and functions on $\SO(3)$ are now identified with functions on $\SU(2)$ invariant under $g \to - g$. The definition of the plane-waves involves a choice of coordinates on the group. Following \cite{aristidedaniele}, we adopt: 
\beq
\forall g \in \SU(2) \,, \qquad \e_g \maps \, x \mapsto e^{\rm{i} \Tr(x |g|)}
\eeq 
where for $g \in \SU(2)$ we denote $|g| \equiv \rm{sign}(\Tr \,g) g$, and $\Tr$ is the trace in the fundamental representation of $\SU(2)$. The Lie algebra variables have a metric interpretation, as vectors associated to the edges of the triangles \cite{aristidedaniele}. The action has the same combinatorial structure as in group variables, except that the pointwise product for functions on $\SU(2)$ is replaced by a non-commutative and non-local product for functions on $\su(2)$, noted $\star$. It is induced by the group structure of $\SU(2)$, as dual to the convolution product for functions on the group. Defined first on plane-waves: 
\beq  \label{star}
(\e_{g} \star \e_{g'})(x) \!:=\! \e_{gg'}(x)\,,
\eeq
it is then extended to the image of the non-commutative Fourier transform by linearity.

We can define the quantum theory by the following path integral:
\beq
\cZ = \int \extd \mu_{inv}(\tvphi_\ell, \overline{\tvphi}_\ell) \, \e^{- S[\tvphi]}\,,
\eeq
but since the Lebesgue measure $\mu_{inv}$ on the space of left invariant fields is not even defined, this is only formal. The strategy usually adopted (a detailed discussion can be found in \cite{scaling3d}) consists in two steps. First the action is re-written in terms of generic fields, with constraints imposing left invariance. Secondly, the non-trivial kinetic term thus obtained is combined with the Lebesgue measure to give a well-defined Gaussian measure. Integrating the exponential of the interaction part of the action with respect to this measure makes sense of the previously ill-defined partition function. This strategy will also be used in the construction of a vertex representation of the model, so let us detail it already in edge variables.

To begin with, we can impose the gauge invariant condition (\ref{gauge}) by group averaging a generic field $\vphi_{\ell} \in L^{2}(\SU(2)^{3})$:
\beq
\tvphi_\ell(g_1, g_2, g_3) = \int \extd h \, \vphi_\ell(hg_1, hg_2, hg_3) \equiv (\cP \act \vphi_\ell)(g_1, g_2, g_3)\,.
\eeq 
In Lie algebra variables, this translates as:
\beq\label{gauge_metric}
\tvphihat_\ell = \widehat{\cP \act \vphi_\ell} = \widehat{C} \star \widehat{\vphi}_\ell\, ,
\eeq
with:
\bes
\widehat{C}(x_1, x_2, x_3) &\equiv& \delta_{0}(x_1 + x_2 + x_3)\\ 
\delta_{x}(y) &\equiv& \int \extd h \, e_{g^{-1}}(x) e_{g}(y) \,.
\ees
The functions $\delta_{x}$ play the role of Dirac distributions in the sense that
\beq
\int \extd y \, (\delta_{x} \star f)(y) = f(x)
\eeq 
for any function $f$ in the image of the non-commutative Fourier transform. Thus, we see that the gauge invariance of the GFT field translates into the closure of the triangle corresponding to it, ensuring geometricity, in accordance with and in confirmation of the interpretation of the Lie algebra variables as edge vectors.

Writing the action in terms of the unconstrained fields $\vphi_\ell$, we notice that the projector $\cP$ induces a non-trivial kinetic term. It will therefore play the role of propagator at the quantum level. Explicitly, the partition function is defined with respect to the Gaussian measure $\mu_{\cP}$ of covariance $\cP$, or its equivalent in metric variables. Namely:
\beq
\cZ \equiv \int \extd \mu_{\cP}(\vphi_\ell, \overline{\vphi}_\ell) \, \e^{- S_{int}[\cP \act \vphi]} = \int \extd \mu_{\widehat{C}}(\vphihat_\ell, \overline{\vphihat}_\ell) \, \e^{- S_{int}[\widehat{C} \star \vphihat]}\,.
\eeq

This partition function generates amplitudes labelled by colored graphs, which we describe in the following. They are made of two types of four valent nodes, which we will call clockwise and anticlockwise, and graphically represent by black and white dots respectively. Their lines have colors $\ell \in \{1\cdots 4\}$, and on each node meet four lines with different colors. Moreover a line is always required to link a clockwise node to an anticlockwise one, so that the graph has no tadpole \cite{vincentcolored}. These conditions ensure that colored graphs are dual to simplicial complexes which triangulate orientable pseudo-manifolds \cite{Gurau:2010nd}, and this is the reason why these objects are well-known in combinatorial topology (see \cite{francesco} and references therein). In this picture lines are dual to triangles, and nodes to tetrahedra of the triangulation. Connected components of the graph made of lines of two different colors are dual to edges, whereas the vertices of the simplicial complex are obtained from the connected components made of three different colors. In GFT this combinatorial structure is usually encoded in a stranded substructure, to which geometrical variables are attached. More precisely, a line of color $\ell$, which represents the propagation of a field of color $\ell$, is made of three strands. These strands are themselves dual to the edges of the triangle the field represent, and we attach to them the corresponding group or Lie algebra variables. These strands are finally paired in nodes, following the pattern of a tetrahedron, as represented in figure (\ref{interaction_edge}). In this picture edges of the triangulation are dual to closed chains of strands. We give the simple example of the so-called \textit{sunshine graph} in figure (\ref{example_edge}). Its dual triangulation is made of two tetrahedra whose faces are identified pairwise, and has the topology of a sphere. 

The amplitude of a given graph $\cG$ can be given several interpretations, depending on the representation one chooses to work with. In metric variables, it has been shown in \cite{aristidedaniele} that it takes the form of a topological simplicial gravity path integral on the simplicial complex dual to $\cG$. In group variables, the amplitude is that of a gauge theory on the dual $2$-complex, imposing flatness of the gauge connection. Finally, we can obtain a third picture by expanding the functions on the group in irreducible representations using harmonic analysis. The amplitude of $\cG$ takes the form of a spin foam model, from which we can make contact with quantum geometry and Loop Quantum Gravity. We refer to the literature for more details \cite{gftreview}.

\begin{figure}[h]
\begin{center}
\includegraphics[scale=0.5]{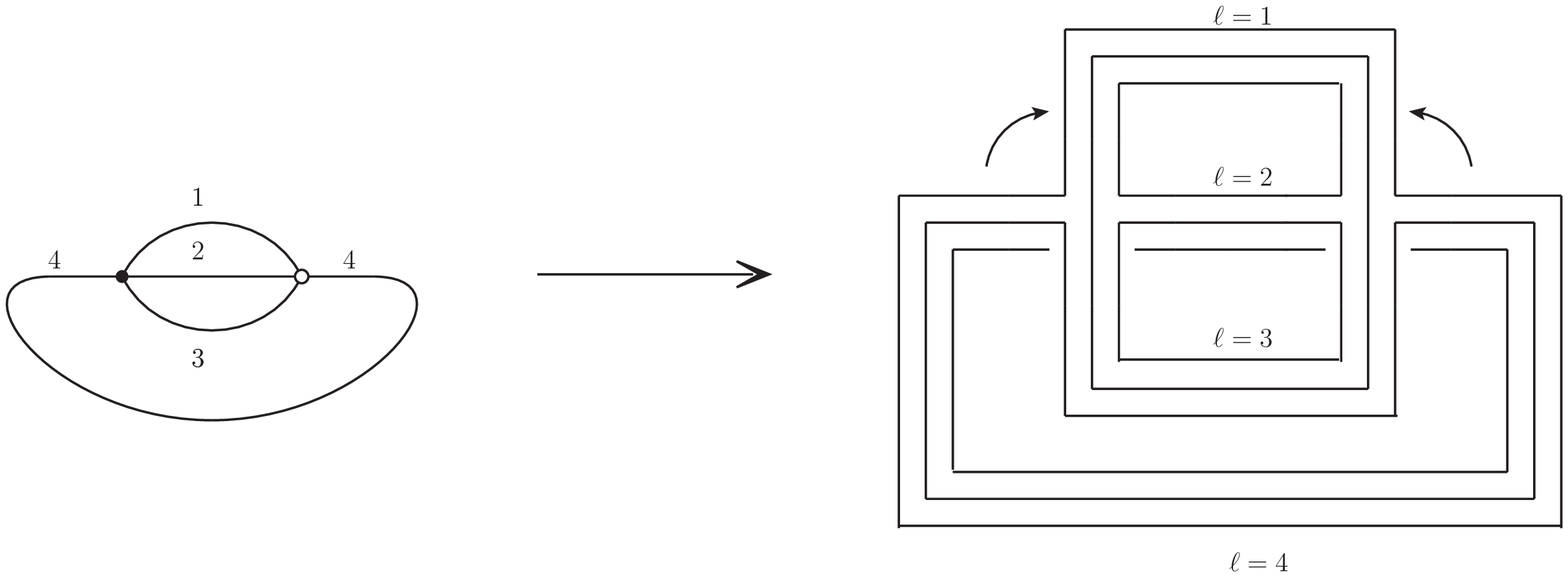} 
\caption{Combinatorial structure of the \textit{sunshine graph} in edge variables.} \label{example_edge}
\end{center}
\end{figure}

\subsection{Translation symmetries}

In the recent work \cite{diffeos}, the model was shown to respect (quantum) symmetries, given by actions of the Drinfel'd double $\cD \SO(3)\!=\!\cc(\SO(3))\rtimes \C\SO(3)$ on the fields. We will focus on the translational part of these actions, interpreted as (discrete) diffeomorphisms \cite{qRC,laurentdiffeo,biancadiffeos,biancabenny}.
They have four generators $\{\cT^{\ell'} , \ell'=1\cdots 4\}$, each $\cT^{\ell'}$ acting non-trivially on fields of color $\ell \neq \ell'$. For instance, $\cT^{3}$ acts on $\tvphi_1$ as:
\beq
\cT^{3}_{\varepsilon} \act \tvphi_1(g_1, g_2, g_3) \equiv (\e_{g_1^\inv} \star \e_{g_3})(\varepsilon) \, \tvphi_1(g_1, g_2, g_3) \, = \e_{g_1^\inv g_3}(\varepsilon) \, \tvphi_1(g_1, g_2, g_3).
\eeq

This can be interpreted as translations of the edges $1$ and $3$, respectively by $\varepsilon$ and $-\varepsilon$, with a deformation given by the $\star$-product. This is clearer in metric variables, where the previous equation can be (schematically) written as:
\beq
\cT^{3}_{\varepsilon} \act \tvphihat_1(x_1, x_2, x_3) = \bigstar_{\varepsilon} \, \tvphihat_1(x_1-\varepsilon, x_2, x_3+\varepsilon)\,. 
\eeq
As a result, the action of $\cT^{3}$ on the field of color $1$ can geometrically be interpreted as a deformed translation of one of its vertices, as represented in figure (\ref{trans_vertex}). Furthermore, we can assign colors to the vertices of the tetrahedron defining the interaction term, with the convention that $v_\ell$ should be the vertex opposed to the triangle of color $\ell$. This induces a color label for vertices of the different triangles. In this picture, the action of $\cT^{3}$ on $\tvphi_1$ corresponds to a translation of the vertex of color $3$ in the triangle of color $1$.

\begin{figure}[h]
\begin{center}
\includegraphics[scale=0.5]{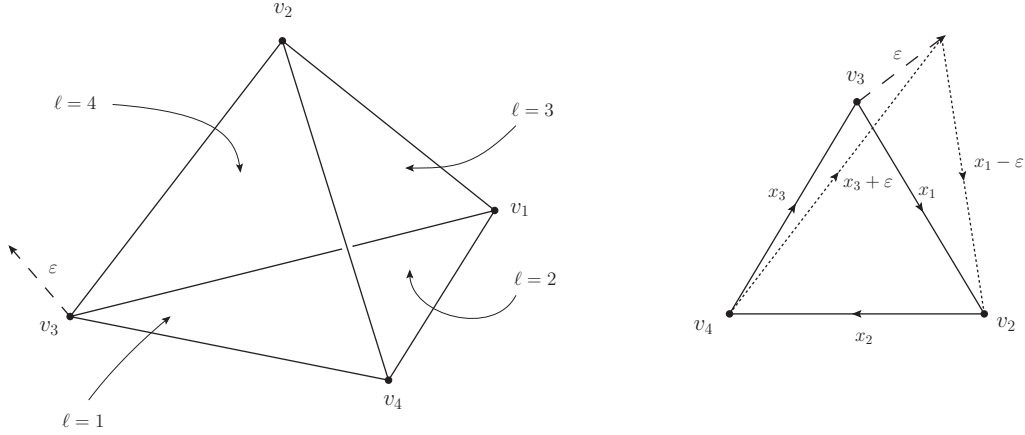} 
\caption{Action of $\cT^{3}_\epsilon$ on the interaction term, and resulting transformation of $\tvphi_1$.} \label{trans_vertex}
\end{center}
\end{figure}

This geometrical interpretation generalizes to any generator and any field: $\cT^{\ell'}_\varepsilon$ translates the vertex of color $\ell'$ in $\tvphi_\ell$ (if any) by a quantity $\varepsilon$. With our conventions, the symmetries are therefore given by the following equations:
\bes \label{VertexTranslation1}
\cT^{1}_{\varepsilon} \act \tvphi_1(g_1, g_2, g_3) &:=&  \tvphi_1(g_1, g_2, g_3) \nn \\
\cT^{1}_{\varepsilon} \act \tvphi_2(g_3, g_5, g_4) &:=& \e_{g_4^\inv g_5}(\varepsilon) \, \tvphi_2(g_3, g_5, g_4) \nn \\
\cT^{1}_{\varepsilon} \act \tvphi_3(g_5, g_2, g_6) &:=& \e_{g_5^\inv g_6}(\varepsilon) \, \tvphi_3(g_5, g_2, g_6) \nn \\
\cT^{1}_{\varepsilon} \act \tvphi_4(g_4, g_6, g_1) &:=& \e_{g_6^\inv g_4}(\varepsilon) \, \tvphi_4(g_4, g_6, g_1) \nn 
\ees
\bes
\cT^{2}_{\varepsilon} \act \tvphi_1(g_1, g_2, g_3) &:=& \e_{g_2^\inv g_1}(\varepsilon) \, \tvphi_1(g_1, g_2, g_3) \nn \\
\cT^{2}_{\varepsilon} \act \tvphi_2(g_3, g_5, g_4) &:=&  \tvphi_2(g_3, g_5, g_4) \nn \\
\cT^{2}_{\varepsilon} \act \tvphi_3(g_5, g_2, g_6) &:=& \e_{g_6^\inv g_2}(\varepsilon) \, \tvphi_3(g_5, g_2, g_6) \nn \\
\cT^{2}_{\varepsilon} \act \tvphi_4(g_4, g_6, g_1) &:=& \e_{g_1^\inv g_6}(\varepsilon) \, \tvphi_4(g_4, g_6, g_1) \nn 
\ees
\bes
\cT^{3}_{\varepsilon} \act \tvphi_1(g_1, g_2, g_3) &:=& \e_{g_1^\inv g_3}(\varepsilon) \, \tvphi_1(g_1, g_2, g_3) \nn \\
\cT^{3}_{\varepsilon} \act \tvphi_2(g_3, g_5, g_4) &:=& \e_{g_3^\inv g_4}(\varepsilon) \, \tvphi_2(g_3, g_5, g_4) \nn \\
\cT^{3}_{\varepsilon} \act \tvphi_3(g_5, g_2, g_6) &:=&  \tvphi_3(g_5, g_2, g_6) \nn \\
\cT^{3}_{\varepsilon} \act \tvphi_4(g_4, g_6, g_1) &:=& \e_{g_4^\inv g_1}(\varepsilon) \, \tvphi_4(g_4, g_6, g_1) \nn 
\ees
\bes
\cT^{4}_{\varepsilon} \act \tvphi_1(g_1, g_2, g_3) &:=& \e_{g_3^\inv g_2}(\varepsilon) \, \tvphi_1(g_1, g_2, g_3) \nn \\
\cT^{4}_{\varepsilon} \act \tvphi_2(g_3, g_5, g_4) &:=& \e_{g_5^\inv g_3}(\varepsilon) \, \tvphi_2(g_3, g_5, g_4) \nn \\
\cT^{4}_{\varepsilon} \act \tvphi_3(g_5, g_2, g_6) &:=& \e_{g_2^\inv g_5}(\varepsilon) \, \tvphi_3(g_5, g_2, g_6) \nn \\
\cT^{4}_{\varepsilon} \act \tvphi_4(g_4, g_6, g_1) &:=& \tvphi_4(g_4, g_6, g_1). \nn 
\ees
Note that the Hopf algebra deformations (i.e. the $\star$-products) are defined such that the plane-waves generating the translations are always of the form $\e_{g_i^\inv g_j}(\varepsilon)$. This feature has also a geometrical meaning: it guarantees that the transformed fields stay invariant under diagonal left action of $\SO(3)$, that is the triangles remain closed after translation of one of their vertices.

To be complete, we would need to specify how these translations act on products of fields. This step, which depends on the $\cD \SO(3)$ coproduct, again amounts to a choice of $\star$-product orderings of the plane waves resulting from the actions on individual fields. One result of \cite{diffeos} is that it is possible to define them in such a way that the action, and in particular its interaction term, are left invariant. We postpone this task to the next section, where the use of vertex variables will make the definitions more geometrically transparent.

The interpretation of these symmetries is very nice. As just mentioned, they are interpreted as translations of the vertices of the triangulation, which at the level of simplicial gravity are the discrete counterparts of the diffeomorphisms \cite{biancadiffeos}. At the discrete gauge field theory level, that is in group space, they impose triviality of the holonomy around a loop encircling a vertex of the boundary triangulation (this is apparent in the group representation of the GFT interaction vertex), which is the content of the diffeomorphism constraints of 3d gravity. Finally, in the spin foam formulation they generate the recurrence relations satisfied by $6j$-symbols \cite{BarrettCraneWdW,eterasimonevalentin}, which again encode the diffeomorphism invariance of the theory in algebraic language and the behaviour under coarse-graining. We refer to \cite{diffeos} for a detailed discussion of these aspects.

%
\section{Transformation to variables associated to vertices} \label{sec:vertex}
%

In this section we explain in details how the action can be reexpressed in terms of fields with vertex variables, as opposed to the usual edge variables. Such a formulation is suggested by the form of the translation symmetries, and makes it easier to analyze them. We will also see that it brings to the forefront some of the topological properties of the simplicial complexes generated by the model in perturbative expansion. This will be the key to the bounds on Feynman amplitudes we will derive in the following sections.

\subsection{From edge to vertex variables}

Following the interpretation of the symmetries as vertex translations, we write each edge Lie algebra variable as a difference between the positions of its two endpoints (with respect to some arbitrary reference point). Each triangle is therefore described by the three positions of its vertices. This amounts to representing quantized triangles by new fields $\tpsi_\ell$, defined as follows:
\bes
\forall \ell \in \left\lbrace 1,..,4\right\rbrace ,\, \tpsi_\ell(u, v, w) &\equiv& \int \extd g_1 \extd g_2 \extd g_3 \tvphi_\ell(g_1, g_2, g_3) \e_{g_2^\inv g_1}(u) \e_{g_1^\inv g_3}(v) \e_{g_3^\inv g_2}(w) \nn \\
&=& \bigstar_u \bigstar_v \bigstar_w \tvphihat_\ell(u - v, w - u , v - w)\,.
\ees
The form of $\tpsi_\ell$ is very specific, and we would not be able to define a measure on this space of fields. As it was already the case in edge variables, in order to define a measure on the space of such fields, a first step towards the definition of the dynamics is to write $\tpsi_\ell$ as a function of some generic field. More precisely, it can be defined after Fourier transform in group representation, in terms of new group variables $G_u := g_2^\inv g_1$, $G_v := g_1^\inv g_3$, $G_w := g_3^\inv g_2$. This gives (using the invariance (\ref{gauge}) of $\tvphi_\ell$):
\bes
\tpsi_\ell(u, v, w) &=& \int \extd G_v \extd G_w \vphi_\ell(G_v^\inv, G_w, \one) \e_{(G_v G_w)^\inv}(u) \e_{G_v}(v) \e_{G_w}(w) \nn \\
		&=& \int \extd G_u \extd G_v \extd G_w \delta(G_u G_v G_w) \vphi_\ell(G_v^\inv, G_w, \one) \e_{G_u}(u) \e_{G_v}(v) \e_{G_w}(w) \nn \\
		&=& \int \extd G_u \extd G_v \extd G_w \delta(G_u G_v G_w) \psi_\ell(G_u, G_v, G_w) \e_{G_u}(u) \e_{G_v}(v) \e_{G_w}(w) \nn \\
		&=& \int \extd G_u \extd G_v \extd G_w \int \extd \varepsilon \bigstar_{\varepsilon} \, \psi_\ell(G_u, G_v, G_w) \e_{G_u}(u + \varepsilon) \e_{G_v}(v + \varepsilon) \e_{G_w}(w + \varepsilon) \nn \\
		&=& \int \extd \varepsilon \bigstar_{\varepsilon} \, \hpsi_\ell(u + \varepsilon, v + \varepsilon, w + \varepsilon).  \nn 
\ees
In the last three lines we introduced an auxiliary field $\psi_\ell$ defined as:
\beq
\forall g_1, g_2, g_3 \in \SU(2), \, \psi_\ell(g_2^\inv g_1, g_1^\inv g_3, g_3^\inv g_2) \equiv \vphi_\ell(g_1, g_2, g_3) ,
\eeq
and the $\star$-products have to be taken in the correct order, namely from left to right.
The last line of the calculation has again a nice geometric interpretation. In usual edge variables, a triangle is specified by three edge vectors which are constrained to close. Alternatively here we give the positions of the vertices up to a global translation, which is irrelevant to the intrinsic geometry of the triangle.
The group variables $G_u$, $G_v$ and $G_w$ are holonomies associated to paths which go from the middle of one edge to the center of the triangle, and then to the middle of another edge. The triangle interpretation of the field requires the triviality of the product $G_u G_v G_w$, as shown in figure (\ref{change_variables}).

\begin{figure}[h]
\begin{center}
\includegraphics[scale=0.5]{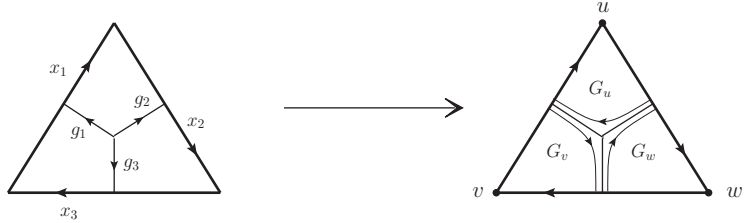} 
\caption{Map from edge to vertex variables.} \label{change_variables}
\end{center}
\end{figure}

\subsection{GFT Action in vertex variables}

\subsubsection{Action in terms of the constrained fields}

As already proven in \cite{diffeos}, the original Boulatov action can be re-written in terms of the new fields $\tpsi_\ell$. With the conventions of this paper, we have:
\bes
S_{kin}[\tpsi] &=& \sum_\ell \int [\extd^3 v_i]^2\, \tpsi_\ell(v_1, v_2, v_3)  \star {\overline{\tpsi}}_\ell(v_1, v_2, v_3) \,, \\
S_{int}[\tpsi] &=& \lambda \int  [\extd^3 v_i]^3\, \tpsi_1(- v_2, v_3, - v_4) \star \tpsi_2(- v_4, v_3, v_1) \star \tpsi_3(- v_4, v_1, - v_2) \star \tpsi_4(v_1, v_3, - v_2) \label{int_vertex} \\
&+& \overline{\lambda} \int  [\extd^3 v_i]^3\,{\overline{\tpsi}}_1(v_2, - v_3, v_4) \star {\overline{\tpsi}}_2(v_4,- v_3,- v_1) \star {\overline{\tpsi}}_3(v_4,- v_1, v_2) \star {\overline{\tpsi}}_4(- v_1,- v_3, v_2) \, .\nn
\ees
We notice that in all the integrals we have one free variable which can be fixed to any value without changing the value of the action; this amounts to a choice of origin from which measuring the position of the vertices. This is also reflected in the four translation symmetries not being independent, one of them being automatically verified when the others are imposed; in other words, the model knows about the intrinsic geometry of the triangles and of the tetrahedra they form, and correctly does not depend on their embedding in $\mathbb{R}^3$. 

We remark also that, in the interaction, each vertex variable appears in three different fields, so that we have a $\star$-product of three terms for each $v_\ell$. The extra signs encode orderings of the $\star$-products, which can again be interpreted as defining the orientations of the triangles. Consider for example the triangle of color $1$. From figure (\ref{trans_vertex}), its orientation is given by the cyclic ordering $(x_1, x_2, x_3)$ of its edge variables, which induces a natural cyclic ordering of its vertices: $(v_2, v_3, v_4)$ (notice that by convention, we actually choose the reverse ordering). This induces in turn an ordering of the triangles attached to the vertex $v_1$: $(\ell = 2, \ell =3, \ell = 4)$ (see again the left part of figure (\ref{trans_vertex})). This is the (cyclic) order in which, in the clockwise interaction term, the $\star_{v_1}$-product of fields having $v_1$ in their arguments (that is $\tpsi_2$, $\tpsi_3$ and $\tpsi_4$) has to be computed. In the anticlockwise interaction term, this has to be reversed. That is why the variable $v_1$ appears with a positive sign in the first interaction term, and a minus sign in the second. This discussion generalizes to any color, so that in the end signs in front of variables $v_\ell$ are fully determined by the ordering of variables in the field $\tpsi_\ell$ of the same color.

\subsubsection{Action in terms of the unconstrained fields}

Anticipating the next section, where we will need a well-defined measure on the space of fields, we now write the same GFT action in terms of unconstrained fields $\widehat{\psi}_\ell$. At this level it is easier to use group variables, that is actually write everything in terms of $\psi_\ell$. A direct computation shows that:
\bes
S[\psi] &=&  \sum_{l} \int [\prod_{i = 1}^{3} \extd G_i \prod_{i = 1}^{3} \extd \tilde{G}_i] \cK(G_i;\tilde{G}_i) \psi_\ell(G_1, G_2, G_3) \overline{\psi}_\ell(\tilde{G}_1, \tilde{G}_2, \tilde{G}_3) \nn \\
&+& \lambda \int [\prod_{\ell \neq \ell'} \extd G^l_{l'}] \cV(G^l_{l'}) \psi_1^{234} \psi_2^{431} \psi_3^{412} \psi_4^{132} \label{action_vertex} \\
&+& \overline{\lambda} \int [\prod_{\ell \neq \ell'} \extd G^l_{l'}] \cV(G^l_{l'}) {\overline{\psi}}_1^{\lower0.1in \hbox{\footnotesize234}} {\overline{\psi}}_2^{\lower0.1in \hbox{\footnotesize431}} {\overline{\psi}}_3^{\lower0.1in \hbox{\footnotesize412}} {\overline{\psi}}_4^{\lower0.1in \hbox{\footnotesize 132}}\nn ,
\ees
with $\psi_{\ell}^{ijk} \equiv \psi_{\ell}(G^{\ell}_i, G^{\ell}_j, G^{\ell}_k)$ and:
\bes \label{Feynmanfunctions}
\cK(G_1, G_2, G_3; \tilde{G}_1, \tilde{G}_2, \tilde{G}_3) &=& \delta(G_{1} G_{2} G_{3}) \delta(G_{1} {\tilde{G}_{1}}^\inv) \delta(G_{2} {\tilde{G}_{2}}^\inv) \delta(G_{3} {\tilde{G}_{3}}^\inv),  \\
\cV(G^l_{l'}) &=& \delta(G^{1}_{2} G^{1}_{3} G^{1}_{4})\delta(G^{2}_{4} G^{2}_{3} G^{2}_{1})\delta(G^{3}_{4} G^{3}_{1} G^{3}_{2})\delta(G^{4}_{1} G^{4}_{3} G^{4}_{2}) \nn \\
		&& \delta(G^{4}_{2} G^{3}_{2} G^{1}_{2})\delta(G^{4}_{3} G^{1}_{3} G^{2}_{3}) \delta(G^{1}_{4} G^{3}_{4} G^{2}_{4}). \label{int_func_vertex}
\ees 

\begin{figure}
\begin{center}
\includegraphics[scale=0.5]{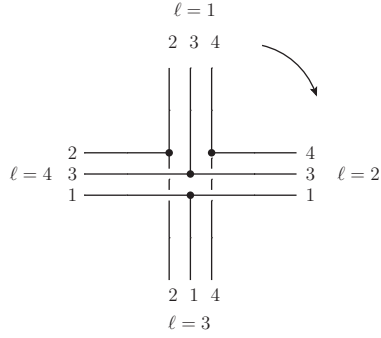} 
\caption{Combinatorics of the interaction function in vertex variables. One of the four three valent interactions is redundant. The arrow indicates the ordering of variables.} \label{interaction_vertex}
\end{center}
\end{figure}

Recall that we previously assigned colors to the vertices in the interaction tetrahedron, with the following simple convention: the vertex $v_\ell$ is the one opposed to the triangle of color $\ell$. We use this labelling in (\ref{int_func_vertex}), where upper indices correspond to colors of the triangles whereas lower ones correspond to that of the vertices. The first four $\delta$-functions come from the translation invariance of the triangles, and the three others encode their gluing through the vertices. In the last line, the fact that there are only three of the four possible $\delta$-functions is again because the four symmetries are not independent. We could alternatively add a $\delta(G^{2}_{1} G^{3}_{1} G^{4}_{1})$ and remove one of the three other $\delta$-functions, a freedom which will prove useful in the computation of the amplitudes. 

For completeness, we finally give the action in the Lie algebra setting:
\bes\label{action_vertex_metric}
S[\hpsi] &=&  \sum_{l} \int [\prod_{i = 1}^{3} \extd v_i \prod_{i = 1}^{3} \extd \tilde{v}_i] \left( \hpsi_\ell(v_1, v_2, v_3) \overline{\hpsi}_\ell(\tilde{v}_1, \tilde{v}_2, \tilde{v}_3)\right)  \star \cK(v_i;\tilde{v}_i) \nn \\
&+& \lambda \int [\prod_{\ell \neq \ell'} \extd v^l_{l'}] \left( \hpsi_1^{234} \hpsi_2^{431} \hpsi_3^{412} \hpsi_4^{132}\right) \star \cV(v^l_{l'}) \nn \\
&+& \overline{\lambda} \int [\prod_{\ell \neq \ell'} \extd v^l_{l'}] \left( {\overline{\hpsi}}_1^{\lower0.1in \hbox{\footnotesize234}} {\overline{\hpsi}}_2^{\lower0.1in \hbox{\footnotesize431}} {\overline{\hpsi}}_3^{\lower0.1in \hbox{\footnotesize412}} {\overline{\hpsi}}_4^{\lower0.1in \hbox{\footnotesize 132}}\right) \star \cV(- v^l_{l'})\nn ,
\ees
with $\hpsi_{\ell}^{ijk} \equiv \hpsi_{\ell}(v^{\ell}_i, v^{\ell}_j, v^{\ell}_k)$, and:
\bes
\cK(v_1, v_2, v_3; \tilde{v}_1, \tilde{v}_2, \tilde{v}_3) &=& \int \extd \varepsilon \left( \delta_0(\varepsilon - v_1 + \tilde{v}_{1}) \star_\varepsilon \delta_0(\varepsilon - v_2 + \tilde{v}_{2}) \star_\varepsilon \delta_0(\varepsilon - v_3 + \tilde{v}_{3}) \right)  \, , \nn \\
\cV(v^l_{l'}) &=& \int [\prod_{\ell = 1}^{4} \extd \varepsilon^{\ell}] \int [\prod_{\ell' = 1}^{4} \extd v_{\ell'}] \, \delta_0(v_{\ell_0}) \nn \\
	&\star_{v_{\ell'}}& \left( \delta_0(\varepsilon^1 - v^{1}_{2} - v_2 ) \star_{\varepsilon^1} \delta_0(\varepsilon^1 - v^{1}_{3} + v_3 ) \star_{\varepsilon^1} \delta_0(\varepsilon^1 - v^{1}_{4} - v_4 ) \right) \nn \\
	&\star_{v_{\ell'}}& \left( \delta_0(\varepsilon^2 - v^{2}_{4}  - v_4 ) \star_{\varepsilon^2} \delta_0(\varepsilon^2 - v^{2}_{3} + v_3 ) \star_{\varepsilon^2} \delta_0(\varepsilon^2 - v^{2}_{1} + v_1 ) \right) \nn \\
	&\star_{v_{\ell'}}& \left( \delta(_0\varepsilon^3 - v^{3}_{4}  - v_4 ) \star_{\varepsilon^3} \delta_0(\varepsilon^3 - v^{3}_{1} + v_1 ) \star_{\varepsilon^3} \delta_0(\varepsilon^3 - v^{3}_{2} - v_2 ) \right) \nn \\
	&\star_{v_{\ell'}}& \left( \delta_0(\varepsilon^4 - v^{4}_{1} + v_1 ) \star_{\varepsilon^4} \delta_0(\varepsilon^4 - v^{4}_{3} + v_3 ) \star_{\varepsilon^4} \delta_0(\varepsilon^4 - v^{4}_{2} - v_2 ) \right) \nn \, .
\ees 
It is obtained from (\ref{action_vertex}) by first expanding the fields $\psi_\ell$ in terms of their Fourier transforms $\hpsi_\ell$ \footnote{The inverse formula is given by
$f(g) = \frac{1}{\pi} \int \extd x (\widehat{f} \star \e_{g^{\inv}}) (x)$ for a function $f$ on $\SO(3)$. We refer to \cite{majidfreidel, aristidedaniele} for details.}, then decomposing the group $\delta$-functions in plane-waves, and finally integrating the holonomies. Now the different gluings are encoded by non-commutative $\delta$-functions on the Lie algebra. For example the propagator determines the gluing of two triangles through their vertices $v$ and $\tilde{v}$, up to a global translation parametrized by $\varepsilon$. Likewise in the interaction, the $\varepsilon^{\ell}$ variables are associated to global translations of the triangles of color $\ell$. As for the variables $v_{\ell'}$, they give $3$-valent interactions on strands of color $\ell'$. Note that the $\delta$-function appearing in the measure has $v_{\ell_0}$ for argument, with $\ell_0$ any of the four colors. This is how the fact that the four $3$-valent interactions are not independent manifests itself: once the triangles have been glued along three of the vertices of the tetrahedron, the fourth gluing is automatic. Finally, signs in front of variables $v_{\ell'}$ implement the correct ordering of $\star$-products, that is the orientations of the triangles.

\subsection{Translation symmetries in vertex variables}

Now we have a vertex representation of the classical theory, it is interesting to revisit and discuss further the translation symmetries. As expected, we have simpler formulas in this representation. 

Let us first discuss the action of translations on individual fields. We can equivalently work in a group or algebra picture, and also with constrained or unconstrained fields, this is irrelevant here. For definiteness we use the fields $\hpsi_\ell$. The transformations read:
\bes \label{VertexTranslation2}
\cT^{1}_{\varepsilon} \act \hpsi_1(v_2, v_3, v_4) &=&  \hpsi_1(v_2, v_3, v_4) \nn \\
\cT^{1}_{\varepsilon} \act \hpsi_2(v_4, v_3, v_1) &=&  \hpsi_2(v_4, v_3, v_1 + \varepsilon) \nn \\
\cT^{1}_{\varepsilon} \act \hpsi_3(v_4, v_1, v_2) &=&  \hpsi_3(v_4, v_1 + \varepsilon, v_2) \nn \\
\cT^{1}_{\varepsilon} \act \hpsi_4(v_1, v_3, v_2) &=&  \hpsi_4(v_1 + \varepsilon, v_3, v_2) \nn 
\ees
\bes
\cT^{2}_{\varepsilon} \act \hpsi_1(v_2, v_3, v_4) &=&  \hpsi_1(v_2 + \varepsilon, v_3, v_4) \nn \\
\cT^{2}_{\varepsilon} \act \hpsi_2(v_4, v_3, v_1) &=&  \hpsi_2(v_4, v_3, v_1) \nn \\
\cT^{2}_{\varepsilon} \act \hpsi_3(v_4, v_1, v_2) &=&  \hpsi_3(v_4, v_1, v_2 + \varepsilon) \nn \\
\cT^{2}_{\varepsilon} \act \hpsi_4(v_1, v_3, v_2) &=&  \hpsi_4(v_1, v_3, v_2 + \varepsilon) \nn 
\ees
\bes
\cT^{3}_{\varepsilon} \act \hpsi_1(v_2, v_3, v_4) &=&  \hpsi_1(v_2, v_3 + \varepsilon, v_4) \nn \\
\cT^{3}_{\varepsilon} \act \hpsi_2(v_4, v_3, v_1) &=&  \hpsi_2(v_4, v_3 + \varepsilon, v_1) \nn \\
\cT^{3}_{\varepsilon} \act \hpsi_3(v_4, v_1, v_2) &=&  \hpsi_3(v_4, v_1, v_2) \nn \\
\cT^{3}_{\varepsilon} \act \hpsi_4(v_1, v_3, v_2) &=&  \hpsi_4(v_1, v_3 + \varepsilon, v_2) \nn 
\ees
\bes
\cT^{4}_{\varepsilon} \act \hpsi_1(v_2, v_3, v_4) &=&  \hpsi_1(v_2, v_3, v_4 + \varepsilon) \nn \\
\cT^{4}_{\varepsilon} \act \hpsi_2(v_4, v_3, v_1) &=&  \hpsi_2(v_4 + \varepsilon, v_3, v_1) \nn \\
\cT^{4}_{\varepsilon} \act \hpsi_3(v_4, v_1, v_2) &=&  \hpsi_3(v_4 + \varepsilon, v_1, v_2) \nn \\
\cT^{4}_{\varepsilon} \act \hpsi_4(v_1, v_3, v_2) &=&  \hpsi_4(v_1, v_3, v_2) \nn 
\ees
Thus each field $\hpsi_\ell$ can be interpreted as living in the representation space of (the translation part of) three copies of the deformed 3d Poincare group $\cD \SO(3)$. This makes the interpretation of these transformations as vertex translations more explicit, and clarifies the very definition of the GFT.

The deformation of the translations manifests itself when acting on product of fields. This is a question we left open in the previous sections, exactly because it is more easily understood in the vertex formulation. To define the action of the translations on a product of fields, we need to interpret it as a tensor product. There is no canonical choice: for example the integrand $\psi_1^{234} \psi_2^{431} \psi_3^{412} \psi_4^{132}$ in the interaction term (\ref{int_vertex}) can be interpreted as the evaluation of $\psi_1^{234} \ot \psi_2^{431} \ot \psi_3^{412} \ot \psi_4^{132}$, but also of $\psi_2^{431} \ot \psi_1^{234} \ot \psi_3^{412} \ot \psi_4^{132}$, and generally of any permutation of the representation spaces. The Hopf algebra deformation of the translations required to make the interaction invariant will then depend on this additional convention. 
For definiteness let us interpret the term $\psi_1^{234} \psi_2^{431} \psi_3^{412} \psi_4^{132}$ as the evaluation of $\psi_1^{234} \ot \psi_2^{431} \ot \psi_3^{412} \ot \psi_4^{132}$. The Hopf algebra structure of the symmetries then has to be consistent with orderings of $\star$-products (i.e. signs) in equation (\ref{int_vertex}). This requires to distinguish colors $\{1,3\}$ from $\{2,4\}$, since the corresponding variables have opposite signs in (\ref{int_vertex}). 
All this suggests the following definition of translations, on products of fields, which we give in group variables. If $\{\phi_i, \, i =1, \cdots, N\}$ are living in the representation space of $\cT^{\ell}$, then:
\bes
\cT^{\ell}_{\varepsilon} \act (\phi(g_1) \ot \cdots \ot \phi(g_N)) &\equiv& \e_{g_1 \cdots g_N} (\varepsilon) (\phi(g_1) \ot \cdots \ot \phi(g_N))\,, \qquad \rm{if} \; \ell \in \{1, 3 \} \\
\cT^{\ell}_{\varepsilon} \act (\phi(g_1) \ot \cdots \ot \phi(g_N)) &\equiv& \e_{g_N \cdots g_1} (\varepsilon) (\phi(g_1) \ot \cdots \ot \phi(g_N))\,, \qquad \rm{if} \; \ell \in \{2, 4 \}\,.
\ees 
With this definition, and the tensor product interpretation of the interaction term we gave, the action is indeed invariant under translations. For instance, in metric variables, the integrand of the interaction part of the action is simply translated with respect to its variable of color $\ell$ under the transformation $\cT^{\ell}$. As a result, and because it is defined by integrals over the whole space $\su(2)$, the invariance follows\footnote{The theory can be made to be independent of any such choice, if one introduce an appropriate non-trivial braiding among GFT fields, which intertwines the translation symmetry; this issue has been raised already in \cite{diffeos} (and in \cite{eteraflorian}) and it is currently under investigation.}.

%
\subsection{Quantum theory}
%

In this section we discuss the path integral quantization of the model, via its perturbative expansion in Feynman diagrams. The partition function of the Boulatov model being divergent, we will regularize it by introducing a suitable cut-off already at the level of the action.

\subsubsection{Cut-off and rescaling of the fields}
There are of course several ways of regularizing the action. Here we choose to regularize the $\delta$-functions to $\delta^{\Lambda}$, $\Lambda$ being a sharp large spin cut-off in the harmonic expansion of $\delta$. The reason is that we would like to make contact with all the bounds we know for the amplitudes of the theory, which in particular allowed to define a $1/N$ expansion \cite{RazvanN, RazvanVincentN, RazvanFullN}. An alternative would be to use a heat kernel regularization, as for example in \cite{ValentinMatteo1, ValentinMatteo2, ValentinMatteo3}. The latter should be better adapted to the metric variables, since it is equivalent to a regularization of the $\su(2)$ integrals compatible with the $\star$-product (the heat kernel becomes a non-commutative gaussian function in Lie algebra space). However for actual computations, the group picture looks better suited, and the large spin cut-off is very natural.

In addition, we rescale the fields and the coupling constant as:
\bes
\psi_\ell \mapsto \frac{\psi_\ell}{\sqrt{\delta^\Lambda(\one)}}  \\
\lambda \mapsto \frac{\lambda}{\sqrt{\delta^\Lambda(\one)}}.\label{rescaling_coupling}
\ees
The first ensures that the kinetic function defines a projector, and the second allows to obtain a uniform degree of divergence for the maximally divergent graphs at all orders \cite{RazvanN}. The cut-off kinetic and interaction functions we will use are then:
\bes
\cK^{\Lambda}(G_1, G_2, G_3; \tilde{G}_1, \tilde{G}_2, \tilde{G}_3) &=& \frac{\delta^\Lambda(G_{1} G_{2} G_{3})}{\delta^\Lambda(\one)} \delta^\Lambda(G_{1} {\tilde{G}_{1}}^\inv) \delta^\Lambda(G_{2} {\tilde{G}_{2}}^\inv) \delta^\Lambda(G_{3} {\tilde{G}_{3}}^\inv),  \nn \\
\cV^{\Lambda}(G^l_{l'}) &=& (\delta^\Lambda(\one))^{\frac{3}{2}} \frac{\delta^\Lambda(G^{1}_{2} G^{1}_{3} G^{1}_{4})}{\delta^\Lambda(\one)} \frac{\delta^\Lambda(G^{2}_{4} G^{2}_{3} G^{2}_{1})}{\delta^\Lambda(\one)} \frac{\delta^\Lambda(G^{3}_{4} G^{3}_{1} G^{3}_{2})}{\delta^\Lambda(\one)} \frac{\delta^\Lambda(G^{4}_{1} G^{4}_{3} G^{4}_{2})}{\delta^\Lambda(\one)} \\
		&& \delta^\Lambda(G^{4}_{2} G^{3}_{2} G^{1}_{2})\delta^\Lambda(G^{4}_{3} G^{1}_{3} G^{2}_{3})\delta^\Lambda(G^{1}_{4} G^{3}_{4} G^{2}_{4}). \nn
\ees 

\subsubsection{Partition function and Feynman rules}
As pointed out in \cite{scaling3d} one can adopt the same strategy as in usual quantum field theory, and make sense of the partition function as a well-defined Gaussian integral. Roughly, the ill-defined Lebesgue measure on the space of fields is combined with the kinetic function, giving an integral of the exponential of the interaction term with respect to a Gaussian measure whose covariance is the kinetic function:
\beq\label{path_int_v}
\cZ^{\Lambda} = \int \extd \mu_{\cK^{\Lambda}}(\psi_\ell, \overline{\psi}_\ell) \, \e^{- S_{int}[\psi]}.
\eeq

The amplitudes are given by the usual colored graphs. What differs is that now the stranded structure is associated to the vertices of the corresponding triangulation. The propagator being a projector we can discard its contribution to the interaction function. We end up with the Feynman rules represented in (\ref{feynman_rules}).
\begin{figure}[h]
\begin{center}
\includegraphics[scale=0.5]{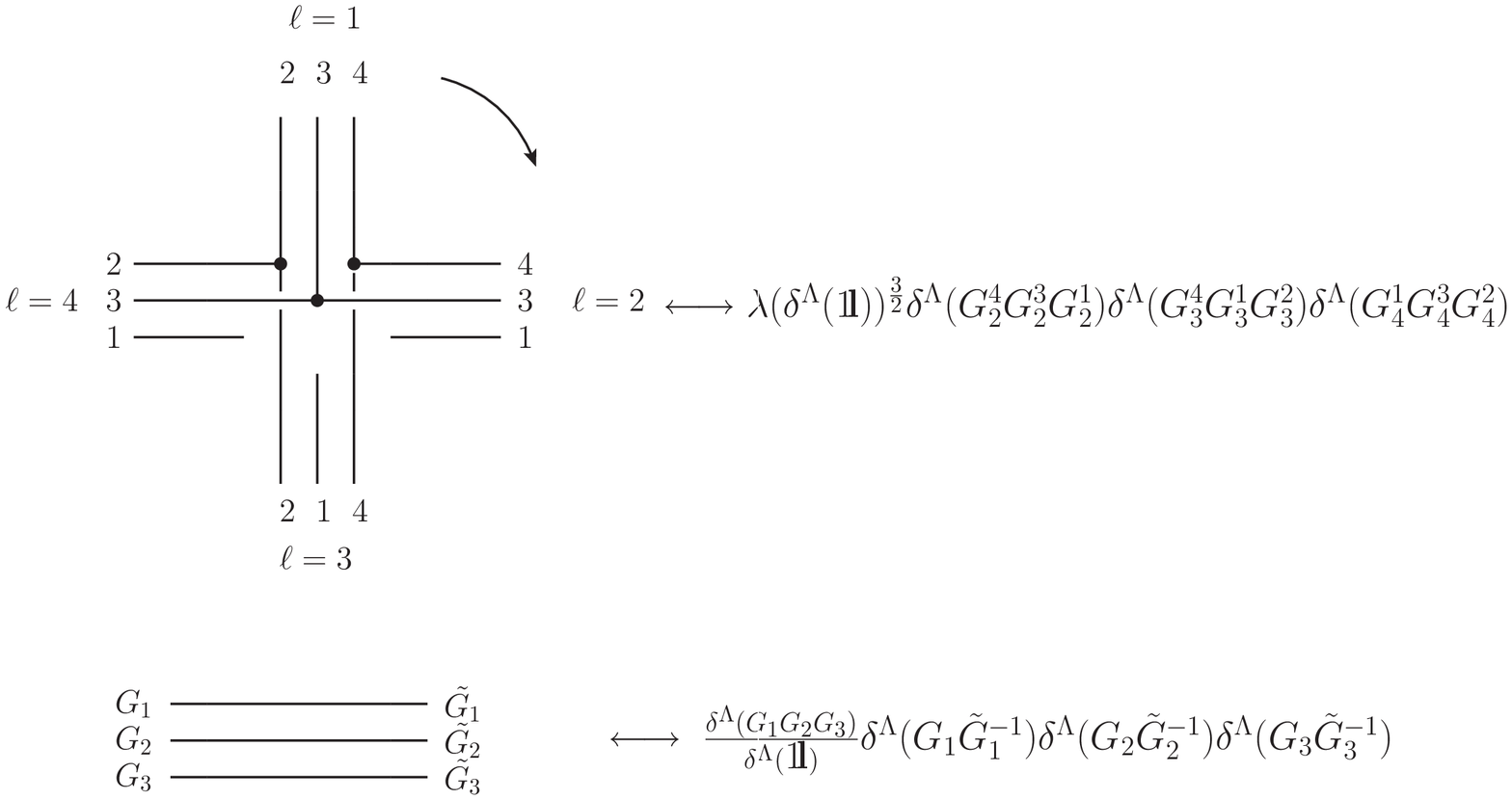} 
\caption{Feynman rules for the clockwise vertex and the propagators. In the second picture, as a matter of convention, we attached the left part of the propagator to a clockwise interaction vertex, and the right part to an anticlockwise one.} \label{feynman_rules}
\end{center}
\end{figure}

At this point we would like to stress again that the path integral is strictly the same than the usual one written in edge variables. The reason why it is so is that at the level of gauge invariant fields, our constructions amounts to a simple (and regular) change of variables in the fields. Therefore the Jacobian of the transformation evaluates to one. This ensures that when the cut-off are removed, the path integrals (and Feynman amplitudes) in edge and vertex variables are the same.
Moreover, we could equally do the very same construction, starting instead from the regularized path integral in edge variables. Knowing the expression for the change of variables in group picture, we could even avoid using metric variables altogether. Starting with regularized $\delta^{\Lambda}$-functions in the path integral in edge variables, we would end up with the regularized path integral (\ref{path_int_v}) in vertex variables. This is ensured by the properties of the regularized $\delta^{\Lambda}$-functions, which behave as Dirac distributions on the space of fields with cut-off $\Lambda$. As a consequence, the vacuum amplitudes in vertex variables will be the same as in edge variables, since in both cases they are coefficients in the perturbative expansion of the partition function in $\lambda$. We would need more care in the case of open graphs, since this would require to match boundary states in both pictures. We do not address this issue in this paper, and focus on vacuum amplitudes in the following.

%
\section{Quantum amplitudes} \label{sec:bubble}
%

In this section we focus on the Feynman amplitudes of the model, by first looking at the combinatorics of the graphs in vertex variables. Then, we will factorize the amplitudes in contributions coming from bubbles of a given color. This will be the starting point for our derivation of new scaling bounds depending on the topology of bubbles only, thus characterizing the manifold-ness of the corresponding simplicial complexes.

\subsection{Explicit bubble structure}

The first thing to notice at this point is that the (contributions to the full) GFT interaction associated to a single vertex are 3-valent, and only strands with the same (vertex) color can interact. Each of them encodes the gluing of three triangles on a vertex of the triangulation, or alternatively the triviality of the holonomy around the wedge associated to this vertex (see figure \ref{flatness}). 
Now look at a connected component of the subgraph of color $\ell$. This is a graph of a non-commutative $\Phi^3$ scalar field theory on a Lie algebra spacetime $\su(2)$, with momentum space $\SU(2)$, the interaction being essentially momentum conservation at each vertex. From the simplicial perspective, it is dual to a bubble around a vertex of color $\ell$, and encodes its topological structure. Precisely, each line of this 3-graph has a color: that of the 4-graph line it is part of. Thus we really have a colored 3-graph, which is therefore dual to a closed and orientable triangulated surface \cite{Vince,lrd,diffeos}:  the bubble.
The overall amplitude associated to a 4-graph is therefore given by $\Phi^3$ graphs encoding the structure of the bubbles, glued to one another through propagators (associated to triangles). We give a simple example in figure (\ref{example_graph}).

\begin{figure}[h]
\begin{center}
\includegraphics[scale=0.5]{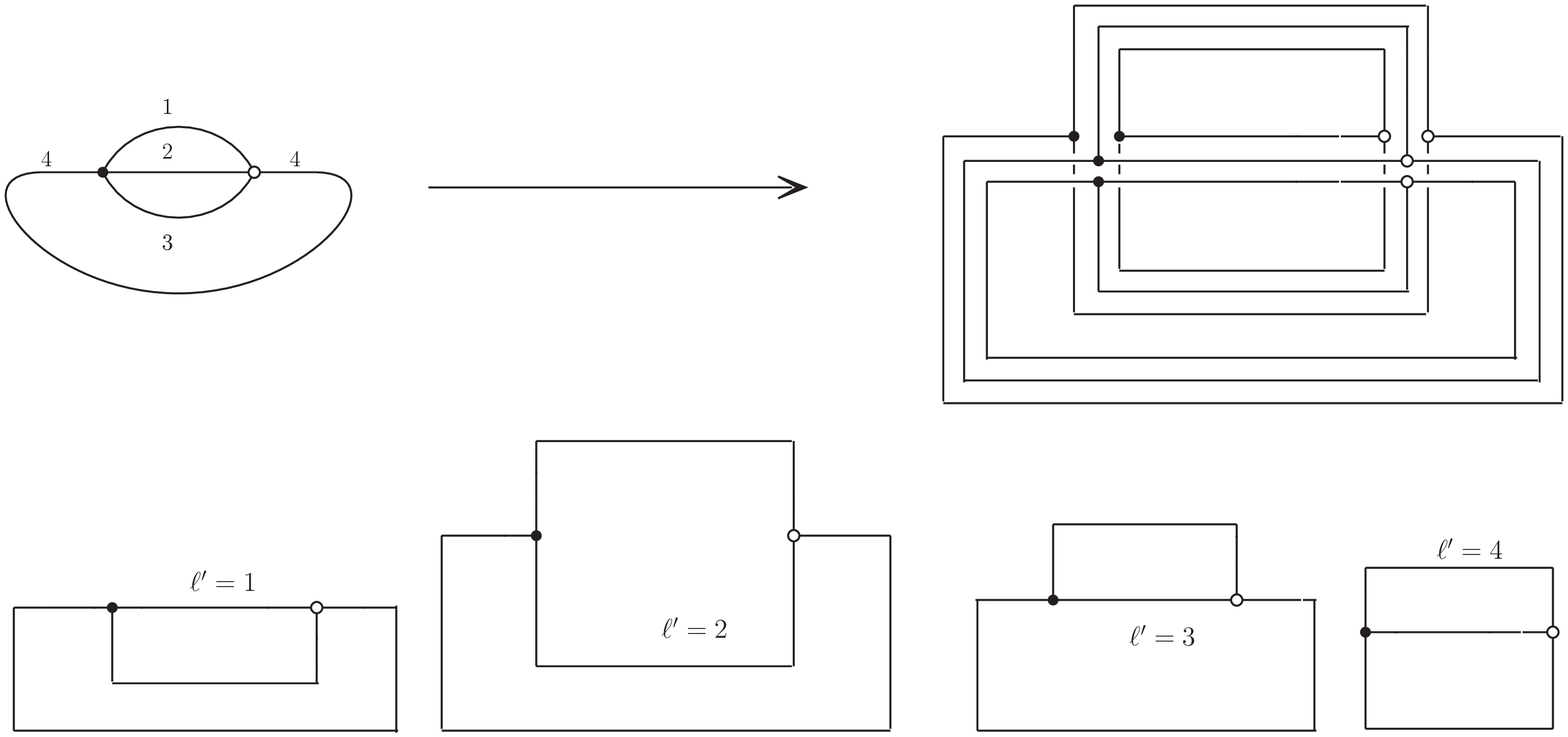} 
\caption{Combinatorial structure of the sunshine graph in vertex variables, and its four bubble graphs.} \label{example_graph}
\end{center}
\end{figure}

\subsection{Flatness of the triangulation}

Before writing the amplitudes in a nice compact form, let us see what the geometrical meaning of the different terms is. In figure (\ref{flatness}) we represent a tetrahedron of the triangulation, dual to an interaction vertex. The elementary variables are holonomies on paths around the vertices of the different triangles, as already shown in figure (\ref{change_variables}). We have two types of constraints on these variables, coming respectively from the propagators and the stranded interactions. The first set of constraints ensure the fields can indeed be interpreted as triangles. As for the interaction associated to the gluing on a vertex $v$, it imposes flatness of the surface formed by the three wedges around $v$.

\begin{figure}[h]
\begin{center}
\includegraphics[scale=0.5]{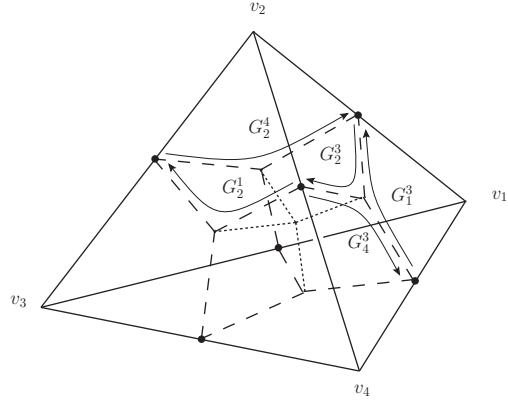} 
\caption{Tetrahedron dual to an interaction vertex. The amplitude imposes two kind of conditions: consistency conditions on triangles, for instance $G^{3}_4 G^{3}_1 G^{3}_2 = \one$ in the triangle of color $3$; and flatness conditions around vertices, for example $G^{4}_2 G^{3}_2 G^{1}_2 = \one$ around the vertex $v_2$.} \label{flatness}
\end{center}
\end{figure}

Before moving on, we give a final confirmation that we recover the usual interpretation of BF amplitudes as measuring moduli spaces of flat connections, in terms of flatness of holonomies around closed paths dual to edges of the simplicial complex. To do so, we transform back to edge variables at the level of the amplitudes, and prove that holonomies around edges of the triangulation are trivial. We give a sketchy proof for a single tetrahedron with boundaries, by showing that its six wedges are flat. The triviality of the propagator ensures that this translates into flatness of face holonomies built by gluing wedges together.

Consider a single tetrahedron with boundary holonomies $\{G^{\ell}_{\ell'}, \ell \neq \ell'\}$ in vertex variables, with the same notations as before. They verify constraints ensuring the triangle interpretation of the boundary fields: for instance $G^{1}_2 G^{1}_3 G^{1}_4 = \one$. As a consequence there exist variables $g^{1}_1$, $g^{1}_2$ and $g^{1}_3$ such that:
\beq
G^{1}_2 = (g^{1}_{2})^{\inv} g^1_1\,, \qquad G^{1}_3 = (g^{1}_{1})^{\inv} g^1_3\,, \qquad G^{1}_4 = (g^{1}_{3})^{\inv} g^1_2\,,
\eeq 
and similarly for the other colors. Note that we label the edges of the tetrahedron (lower indices of $g$ variables) as in figure (\ref{tetrahedron_edge}). Discarding global prefactors, the amplitude is a product of three $\delta$-functions encoding flatness around the vertices $v_1$, $v_2$ and $v_3$:
\bes\label{ampl_vertex}
\cA &\propto& \delta^\Lambda(G^{4}_{2} G^{3}_{2} G^{1}_{2})\delta^\Lambda(G^{4}_{3} G^{1}_{3} G^{2}_{3})\delta^\Lambda(G^{1}_{4} G^{3}_{4} G^{2}_{4}) \\
&\propto& \delta^\Lambda((g^{4}_{1})^{\inv} g^{4}_{6} (g^{3}_{6})^{\inv} g^{3}_{2} (g^{1}_{2})^{\inv} g^1_1)\delta^\Lambda((g^{4}_{4})^{\inv} g^{4}_{1} (g^{1}_{1})^{\inv} g^1_3 (g^{2}_{3})^{\inv} g^{2}_{4})\delta^\Lambda((g^{1}_{3})^{\inv} g^1_2 (g^{3}_{2})^{\inv} g^{3}_{5} (g^{2}_{5})^{\inv} g^{2}_{3})\,.
\ees
On the other hand in usual edge variables, the amplitude is given by:
\bes\label{ampl_edge}
\tilde{\cA} = \int \extd h_1 \extd h_2 \extd h_3 \extd h_4 &&\delta^\Lambda(h_1 g^{1}_{1} (g^{4}_{1})^{\inv} h_4^{\inv}) \delta^\Lambda(h_4 g^{4}_{6} (g^{6}_{3})^{\inv} h_3^{\inv}) \delta^\Lambda(h_3 g^{3}_{2} (g^{1}_{2})^{\inv} h_1^{\inv}) \nn \\ 
&& \delta^\Lambda(h_1 g^{1}_{3} (g^{2}_{3})^{\inv} h_2^{\inv}) \delta^\Lambda(h_3 g^{3}_{5} (g^{2}_{5})^{\inv} h_2^{\inv}) \delta^\Lambda(h_4 g^{4}_{4} (g^{2}_{4})^{\inv} h_2^{\inv})\,, 
\ees
where for each color $\ell$, $h_{\ell}$ is interpreted as the holonomy from the center of the tetrahedron to the center of the triangle of color $\ell$. Thus the amplitude encodes flatness of the six wedges of the tetrahedron. To show that the amplitudes (\ref{ampl_vertex}) and (\ref{ampl_edge}) are equivalent as they should, we can first integrate the auxiliary holonomies $h_2$, $h_3$ and $h_4$ in (\ref{ampl_edge}). The result does not depend on $h_1$ so the last integral is trivial, and we get:
\beq
\tilde{\cA} = \delta^\Lambda(g^{1}_{1} (g^{4}_{1})^{\inv} g^{4}_{6} (g^{3}_{6})^{\inv} g^{3}_{2} (g^{1}_{2})^{\inv})
		\delta^\Lambda(g^{1}_{1} (g^{4}_{1})^{\inv} g^{4}_{6} (g^{3}_{6})^{\inv} g^{3}_{5} (g^{2}_{5})^{\inv} g^{2}_{3} (g^{1}_{3})^{\inv})
		\delta^\Lambda(g^{1}_{1} (g^{4}_{1})^{\inv} g^{4}_{4} (g^{2}_{4})^{\inv} g^{2}_{3} (g^{1}_{3})^{\inv})\,.
\eeq
But this is the same distribution as (\ref{ampl_vertex}), as it can be easily verified. This shows that amplitudes in vertex and edge variables are equal, and both encode flatness of the triangulation.

\subsection{``Reduced'' amplitude}

\begin{figure}[h]
\begin{center}
\includegraphics[scale=0.5]{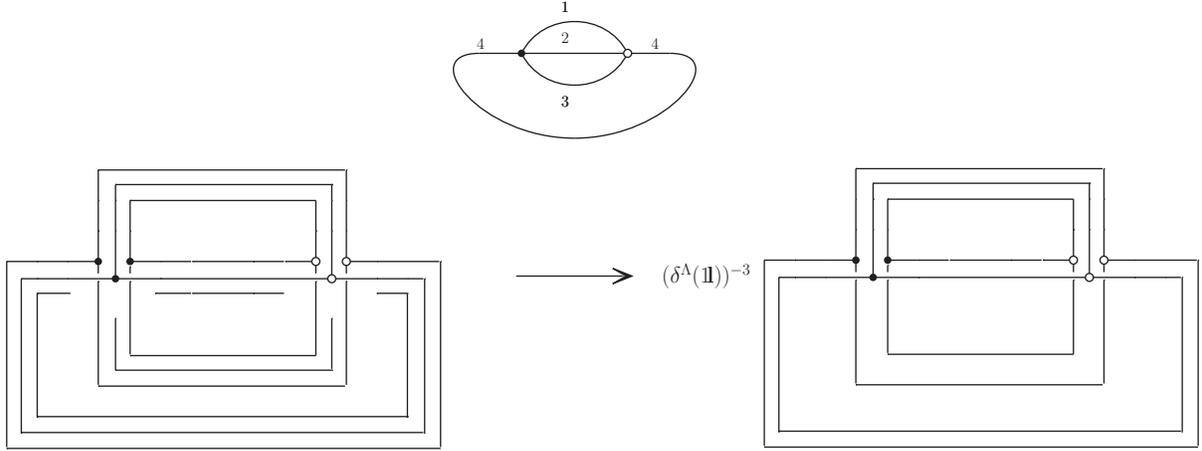} 
\caption{Integration of propagators of colors $2$, $3$ and $4$ in the sunshine graph. A prefactor $(\delta^\Lambda(\one))^{-\frac{3}{2}}$ can be absorbed in each vertex function, which gives the modified Feynman rule shown in figure (\ref{feynman_rules_reduced}).} \label{example_reduced_amplitude}
\end{center}
\end{figure}

Let us consider a graph $\cG$ and compute its amplitude. In each vertex we are free to choose the color of the three strands which do not interact. A simple choice is to pick up the same color in every vertex, say $1$. Now because we have no interaction for strands of color $1$, we can use the variables attached to them to integrate all the propagators of colors $2$, $3$ and $4$ (see figure (\ref{example_reduced_amplitude}) for an example). Each of them gives a contribution $(\delta^\Lambda(\one))^{-1}$ coming from the normalization of the propagator. We can absorb these factors in the interaction terms. Since each integrated line is shared by two interaction vertices, and each interaction vertex has three integrated lines, this amounts to a rescaling of the interaction function by a factor $(\delta^\Lambda(\one))^{-\frac{3}{2}}$. The Feynman rule for the reduced interaction vertex is represented in figure (\ref{feynman_rules_reduced}).

\begin{figure}[h]
\begin{center}
\includegraphics[scale=0.5]{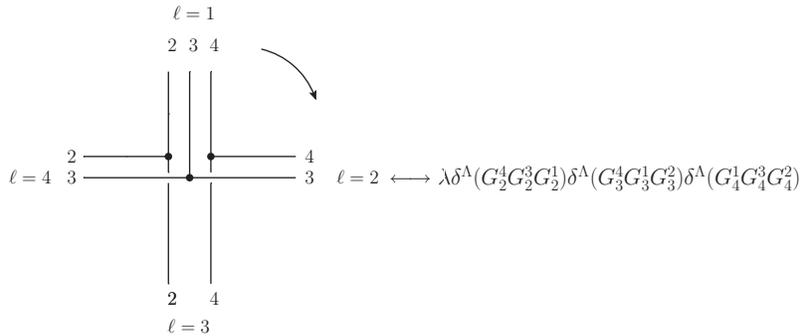} 
\caption{Feynman rule for the reduced clockwise interaction vertex.} \label{feynman_rules_reduced}
\end{center}
\end{figure}

The only propagators left have color $1$, and encode gluings of bubbles of the same color. In each of these bubbles we have now subgraphs of colors $2$, $3$ and $4$ only, which interact through their external strands.

Let us call $\cB_{1}$ the set of bubbles of color $1$. Then for $b \in \cB_{1}$ we note: $V_{b}$, $E_{b}$, $F_{b}$ the sets of vertices, edges and faces of its triangulation; $g_{b}$ its genus; $\cG_{b}$ its (disconnected and open) graph, made of strands of colors $2$, $3$ and $4$. Finally, for $v \in V_{b}$, we define $\triangle^{b}_{v}$ the sets of triangles of $b$ containing $v$. The following combinatorial properties hold:
\begin{enumerate}
 \item $\cG_{b}$ has $3 |F_{b}|$ external strands, and $3 |F_{b}|$ external vertices (all its vertices have to be external since each of them is connected to a propagator of color $1$);
\item Each internal strand of $\cG_{b}$ is dual to one of the end points of an edge of $b$. There are $2 |E_{b}|$ such internal strands;
\item They form connected components of $\cG_{b}$, which are dual to the vertices of $b$. We have therefore $|V_{b}|$ connected components in $\cG_{b}$. The connected component dual to $v \in V_{b}$ has $|\triangle^{b}_{v}|$ strands, which are dual to the triangles of $\triangle^{b}_{v}$.
\end{enumerate}
From this considerations, we see that all the internal strands of $\cG_{b}$ can be integrated. This simply amounts to integrating a loop in a $\Phi^3$ graph, reducing it to one single vertex. We show an example in figure (\ref{strand_reduction}).

\begin{figure}[h]
\begin{center}
\includegraphics[scale=0.5]{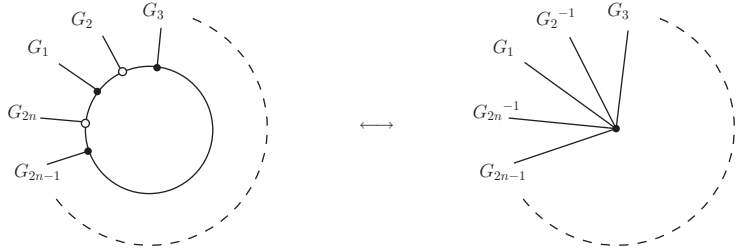} 
\caption{Integration of the internal strands of a connected component of $\cG_{b}$.} \label{strand_reduction}
\end{center}
\end{figure}

Each connected component of $\cG_{b}$ is thus reduced to a unique interaction vertex of valence $|\triangle^{b}_{v}|$, dual to the vertex of the triangulation $v \in V_{b}$. We are left with only one $\delta$-function per vertex of $b$. Defining $\cF_1$ as the set of lines of color $1$, we obtain a general formula for the amplitude:
\beq
\cA^{\cG} = (\lambda \overline{\lambda})^{\frac{\cN }{2}}  \int [\extd G]^{\frac{3 \cN}{2}} \left( \prod_{b \in \cB_{1}} \prod_{v \in V_{b}} \delta^\Lambda\left( \overrightarrow{\prod_{f \in \triangle^{b}_{v}}} (G_{v}^{f})^{\epsilon^{f}_{v}}\right)  \right) \left( \prod_{f \in \cF_{1}} \frac{\delta^\Lambda\left( \overrightarrow{\prod_{v \in f}} G_{v}^{f}\right)}{\delta^\Lambda(\one)}\right) ,
\eeq
where the products of holonomies are ordered according to the initial orientations, and $\epsilon^{f}_{v} = \pm 1$ depending on whether the tetrahedron containing $f$ and $v$ is clockwise or anticlockwise. Since the faces of the bubbles are dual to half lines of color $1$, the following relation holds : $\sum_{b \in \cB_{1}} |F_b| = 2 |\cF_{1}|$. This allows to absorb the normalizations of the remaining propagators into the bubbles:
\beq
\cA^{\cG} = (\lambda \overline{\lambda})^{\frac{\cN }{2}}  \int [\extd G]^{\frac{3 \cN}{2}} \left( \prod_{b \in \cB_{1}} [\delta^\Lambda(\one)]^{- \frac{|F_{b}|}{2}} \prod_{v \in V_{b}} \delta^\Lambda\left( \overrightarrow{\prod_{f \in \triangle^{b}_{v}}} (G_{v}^{f})^{\epsilon^{f}_{v}}\right)  \right) \left( \prod_{f \in \cF_{1}} \delta^\Lambda\left( \overrightarrow{\prod_{v \in f}} G_{v}^{f}\right)\right)\, .
\eeq
But for a bubble $b \in \cB_{1}$ we have the two additional relations:
\bes
2-2 g_b &\equiv& |V_b| -|E_b| + |F_b| \nn \\
3 |F_{b}| &=& 2|E_{b}| \,,
\ees
so that: $-\frac{|F_{b}|}{2} = 2-2 g_b - |V_{b}|$. This gives a formula for the amplitude showing the explicit dependence on the genera of the bubbles of color $1$. Generalizing to any color $\ell$, we obtain the final result of this section:
\bes\label{final_amplitude}
\cA^{\cG} = (\lambda \overline{\lambda})^{\frac{\cN }{2}}  \int [\extd G]^{\frac{3 \cN}{2}} \left( \prod_{b \in \cB_{\ell}} [\delta^\Lambda(\one)]^{2-2 g_b - |V_{b}|} \prod_{v \in V_{b}} \delta^\Lambda\left( \overrightarrow{\prod_{f \in \triangle^{b}_{v}}} (G_{v}^{f})^{\epsilon^{f}_{v}}\right)  \right) \left( \prod_{f \in \cF_{\ell}} \delta^\Lambda\left( \overrightarrow{\prod_{v \in f}} G_{v}^{f}\right)\right)\, .
\ees
The interpretation of the different terms is the following. The pseudo-manifold associated to $\cG$ is now built from $|\cB_{\ell}|$ cells, glued together through their triangulated boundaries (the bubbles). Each bubble $b$ contributes with $|V_b|$ $\delta$-functions imposing flatness of the holonomies around its vertices, plus an overall factor depending on its genus and its number of vertices. These bubbles are glued together through $\frac{\cN}{2}$ $\delta$-functions\footnote{Note that in this sense we can have tadpole lines, which identify two triangles in a same bubble.}, washing out the non geometric data (that is embedding information of the triangles). 

\

Before moving on, let us stress how convenient the above expression for the quantum amplitude is, for analyzing its dependence on the underlying simplicial complex. The first contribution to it (in the first bracket) encodes its degree of manifold-ness entirely, as it only depends on the bubble structure and in a way that is local at the level of each bubble; the second term, on the other hand, encodes the dependence on the structure and topology of the whole complex, but not its possible singular nature, as it depends on how the bubbles, dual to vertices of the simplicial complex, are glued to one another. In the following, we will be only concerned with the bubble structure, and on the issue of the relative suppression of pseudo-manifold configurations over the regular manifolds, and we will thus focus only on the first type of contributions, trying in a sense to \lq trivialize{\rq} the rest of the expression. However, we believe that the above expression could be a natural starting point also for studying the \lq complementary{\rq} issue of the relative weight of different simplicial topologies, for given structure of singularities. We leave this question for future work.

\subsection{First bound}

From the previous formula we can easily derive a general bound for the amplitude of $\cG$. First pick up one strand in each propagator of color $\ell$, and integrate the $\frac{\cN}{2}$ corresponding $\delta$-functions with respect to the variables associated to these strands. This changes the arguments of the $\delta$-functions of the vertices connected to these strands accordingly, possibly in a very complicated way. Anyway, the amplitude as now the form:
\beq
\cA^{\cG} = (\lambda \overline{\lambda})^{\frac{\cN }{2}}  \int [\extd G]^{\cN} \left( \prod_{b \in \cB_{\ell}} [\delta^\Lambda(\one)]^{2-2 g_b - |V_{b}|} \prod_{v \in V_{b}} \delta^\Lambda\left( \cdots \right)  \right)\,,
\eeq
where the dots denote complicated products of holonomy variables. Now using the rough bound $\delta^\Lambda(\cdots) \leq \delta^\Lambda(\one)$for the $|V_b|$ remaining $\delta$-functions  \footnote{This is due to the fact that the $j$-th character of $\SU(2)$ is bounded by its value at the identity $\one$.} per bubble $b$, together with the normalization of the Haar measure, we get:
\beq\label{bound}
\cA^{\cG} \leq (\lambda \overline{\lambda})^{\frac{\cN }{2}}  [\delta^\Lambda(\one)]^{\sum_{b \in \cB_{\ell}} (2 - 2 g_b)} \,.
\eeq

Thus we obtained a bound depending on the topology of the bubbles of color $\ell$. This might seem a bit unnatural, since formula (\ref{bound}) is not symmetric with respect to the color labels. Note however that we have the same kind of bound for any color, and we are completely free to choose any of them. What the refined versions of this inequality will show is exactly that discussing singularities of one given color is relevant, in the sense that it allows to derive optimal bounds.

%
\section{Bounding pseudo-manifolds} \label{sec:bubble2}
%


The efficiency of the bound (\ref{bound}) depends on a competition between the number of bubbles and their topologies. As a consequence, for a given singular topology, the bound can always be made arbitrarily big by adding a large number of planar bubbles (which do not change the singularities). Instead it would be useful to obtain a bound which only depends on the topology of singular bubbles. We thus have to device tools to get rid of as many planar bubbles as possible. The so-called \textit{dipole moves} used in combinatorial topology are of this kind, and were successfully applied to colored group field theory \cite{francesco, RazvanN}. Especially, $1$-dipole contraction moves were used in \cite{RazvanN} to study the large $N$ limit of the Boulatov model (see figure (\ref{dipole})). The author showed that any graph can be successively contracted so as to obtain what is called a \textit{core graph}, with the following property: for any color $\ell$, there is either a unique bubble of color $\ell$, or they are all non-planar. 

\begin{figure}[h]
\begin{center}
\includegraphics[scale=0.5]{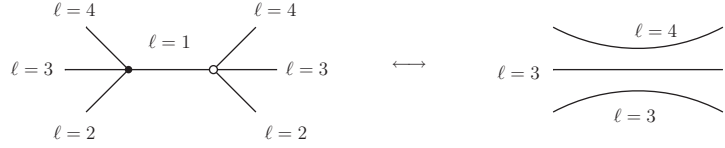} 
\caption{$1$-dipole move. This defines an homeomorphism whenever the two bubbles (of color $1$) connected by the line of color $1$ are different, and at least one of them is planar.} \label{dipole}
\end{center}
\end{figure}

In this section, we will first rederive the 1-dipole contraction move formula, which in our framework amounts to merge two different bubbles. This will then allow us to apply formula (\ref{bound}) to core graphs and obtain sharper bounds on pseudo-manifolds.

\subsection{Reduction to core graphs}

In this paragraph we show that we can derive the properties of 1-dipole contractions from formula (\ref{final_amplitude}). This will guarantee that the bound (\ref{bound}) can be applied to a core graph equivalent to $\cG$ instead of $\cG$ itself. This will imply bounds on pseudo-manifolds, indexed by the type and number of point singularities.

Let us consider two different bubbles $b_1$ and $b_2$ in $\cB_{\ell}$, glued through a triangle $f_{0} \in F_{b_1} \cap F_{b_2}$. We make in addition the assumption that at least one of them, say $b_1$, is a sphere. The contribution of the two bubbles to the amplitude is given by a factor of the form:
\bes
\left( [\delta^\Lambda(\one)]^{2 - |V_{b_1}|} \prod_{v \in V_{b_1}, v \notin f_0} \delta^\Lambda\left( \overrightarrow{\prod_{f \in \triangle^{b_1}_{v}}} (G_{v}^{f})^{\epsilon^{f}_{v}}\right) \right)  \left( [\delta^\Lambda(\one)]^{2-2 g_{b_2} - |V_{b_2}|} \prod_{v \in V_{b_2}, v \notin f_0} \delta^\Lambda\left( \overrightarrow{\prod_{f \in \triangle^{b_2}_{v}}} (G_{v}^{f})^{\epsilon^{f}_{v}}\right)\right) \nn \\
\times \int \extd G_{u_1}^{f_0} \extd G_{u_2}^{f_0} \extd G_{u_3}^{f_0}  \delta^\Lambda\left( G_{u_{1}}^{f_0} G_{u_2}^{f_0} G_{u_3}^{f_0}\right) \prod_{i = 1}^{3} \delta^\Lambda\left( \overrightarrow{\prod_{f \in \triangle^{b_1}_{u_i}}} (G_{u_i}^{f})^{\epsilon^{f}_{v}} \right) \delta^\Lambda\left( \overrightarrow{\prod_{f \in \triangle^{b_2}_{u_i}}} (G_{u_i}^{f})^{\epsilon^{f}_{v}} \right) \, , \nn
\ees
where $u_1$, $u_2$ and $u_3$ are the vertices of $f_0$. Before integrating with respect to $G_{u_{i}}^{f_0}$, we would like to get rid of $\delta^\Lambda\left( G_{u_{1}}^{f_0} G_{u_2}^{f_0} G_{u_3}^{f_0}\right)$, which imposes closure of the triangle $f_0$. Using the other closure and flatness constraints in $b_1$, we see that it is equivalent to saying that the holonomy along a path circling $f_0$ in $b_1$ has to be flat (see figure (\ref{dipole_contraction})). Iterating the process shows that this path can actually be deformed arbitrarily. But $b_1$ is a sphere, hence simply connected. We can therefore contract the path around another triangle of $b_1$, and write the constraint  $G_{u_{1}}^{f_0} G_{u_2}^{f_0} G_{u_3}^{f_0} = \one$ as the closure condition in this triangle. We see thus that $\delta^\Lambda\left( G_{u_{1}}^{f_0} G_{u_2}^{f_0} G_{u_3}^{f_0}\right)$ is redundant and can be set to $\delta^\Lambda\left( \one\right)$ without changing the integral. 

\begin{figure}[h]
\begin{center}
\includegraphics[scale=0.5]{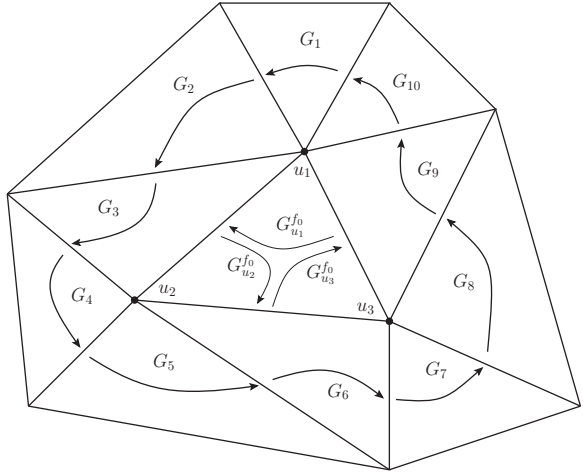} 
\caption{Triangle $f_0$ and its neighbors in $b_1$. Using flatness around $u_1$, $u_2$, $u_3$, and constraints in the three triangles sharing an edge with $f_0$, we show that $G_{u_{1}}^{f_0} G_{u_2}^{f_0} G_{u_3}^{f_0} = G_1 G_2 \cdots G_9 G_{10}$.} \label{dipole_contraction}
\end{center}
\end{figure}

We can now safely integrate the $G_{u_{i}}^{f_0}$ variables, which corresponds to removing $f_{0}$ and taking the connected sum of $b_1$ and $b_2$. We denote this connected sum $b_1 \# b_2$ and refer to \cite{Vince, francesco} for more details. This leads to:
\bes
\left( [\delta^\Lambda(\one)]^{2 - |V_{b_1}| + 2 - |V_{b_2}| - 2 g_{b_2} } \prod_{v \in V_{b_1 \# b_2}} \delta^\Lambda\left( \overrightarrow{\prod_{f \in \triangle^{b_1 \# b_2}_{v}}} (G_{v}^{f})^{\epsilon_{f}}\right) \right)  \times  \delta^\Lambda\left( \one \right) \,.\nn
\ees
But we also have that (the second equality crucially depends on $b_1$ being a sphere):
\bes
|V_{{b_1} \# {b_2}}| = |V_{b_1}| + |V_{b_2}| - 3 \nn \\
g_{{b_1} \# {b_2}} = g_{b_2} \, ,
\ees
so that in the end the contribution of the two bubbles is that of their connected sum. We recover the relation found in \cite{RazvanN}, between the amplitude of the initial graph $\cG$ and the one after absorption of the planar bubble $b_1$ in $b_2$, noted $\cG^{b_{1} \rightarrow b_{2}}$:
\beq
\cA^{\cG} = (\lambda \overline{\lambda}) \cA^{\cG^{b_{1} \rightarrow b_{2}}}
\eeq

This shows that amplitudes are invariant under $1$-dipole contractions, up to $(\lambda \overline{\lambda})$ factors. But it was also shown in \cite{RazvanN} that, for each color, a complete set of $1$-dipole contractions can be performed. We refer to this work for a detailed proof, which relies on a bubble routing. Thus any graph amplitude can be computed from an equivalent graph without $1$-dipoles: a core graph. 

\subsection{A first hierarchy of bounds}

We now use inequality (\ref{bound}) on core graphs\footnote{Strictly speaking, we do not need to contract all the 1-dipoles in the graph, but only that of a given color.} to bound all graphs which are in the same equivalent class under 1-dipole contractions. If $\cG_{p}$ is a core graph with $2p$ vertices, there are two possibilities for its bubbles of color $\ell$: either there is just one of them, or they are all non-planar. If all the bubbles are planar, the graph is dual to an orientable manifold \cite{DP-P,lrd, francesco}. In this case the previous inequality gives, for any color:
\beq
\cA^{\cG_{p}} \leq (\lambda \overline{\lambda})^{p}  [\delta^\Lambda(\one)]^{2} \,. 
\eeq
For a non-manifold core graph, that is when at least one bubble is non-planar (and consequently all the bubbles of the same color), the amplitude is shown to converge, and even to decay to zero as soon as there exists a bubble of genus $2$. More precisely, if $\ell$ is the color of a non-planar bubble, we get:
\bes
\cA^{\cG_{p}} &\leq& (\lambda \overline{\lambda})^{p}  [\delta^\Lambda(\one)]^{{\sum_{b \in \cB_{\ell}}} (2 - 2 g_b)} \nn \\
&\leq& (\lambda \overline{\lambda})^{p}  [\delta^\Lambda(\one)]^{2 - 2 g_{max}} \nn \\
&\leq& (\lambda \overline{\lambda})^{p} \nn ,
\ees
where in the second line $g_{max}$ is defined as the maximal genus in $\cB_{\ell}$. Now remark that the properties of a core graph ensure that ${\sum_{b \in \cB_{\ell}}} (2 - 2 g_b) = {\sum_{b \in \cB^{s}_{\ell}}} (2 - 2 g_b)$, where $\cB^{s}_{\ell}$ is the set of singular bubbles of color $\ell$. But since a core graph has the same singularities as all the graphs which are in the same class, the previous bounds generalize to any graph $\cG$ in the following sense:
\beq\label{bound_core}
\cA^{\cG} \leq (\lambda \overline{\lambda})^{\frac{\cN}{2}}  [\delta^\Lambda(\one)]^{{\sum_{b \in \cB^{s}_{\ell}}} (2 - 2 g_b)}
\leq (\lambda \overline{\lambda})^{\frac{\cN}{2}}  [\delta^\Lambda(\one)]^{2 - 2 g_{max}} \, ,
\eeq
where $\cN$ is the number of nodes of $\cG$, and $g_{max}$ is the maximal genus of its bubbles of color $\ell$. We have thus obtained a hierarchy of bounds, indexed by the types and number of point singularities.

\section{Bounding pseudo-manifolds without reducing to core graphs}\label{sec:optimal}

In this section we study whether the bound (\ref{bound_core}) is optimal or not. We will show that it is, as long as we are concerned with the most degenerate singularities of graphs only. However, if we want to take several singularities into account, we will show that they can be improved. As an interesting by-product, we will see that the reduction to core graphs becomes unnecessary for the purpose of deriving optimal bounds. 

\subsection{Are the bounds optimal?} 

In order to address this question, we need to be able to compute exact amplitudes of a sufficiently rich set of graphs. In this respect, we propose to first design elementary pieces of graphs which have one unique bubble of color $\ell$, and a certain number of external legs. We will then be able to build connected vacuum graphs with any kinds of bubbles out of these elementary graphs. Of course we want to keep the combinatorics of these elementary graphs rather simple, to be able to do exact calculations. 

It is then natural to start from minimal $3$-graphs representing $2$-dimensional orientable surfaces of a given topology. They are called \textit{canonical graphs} in the mathematical literature (see \cite{Vince} and references therein). A canonical graph of genus $g$ has $2(2g+1)$ nodes. Figure (\ref{bubbles}) shows the canonical graphs of genus $0$, $1$ and their generalization to any genus $g$. We refer to \cite{Vince} for proofs of these statements and further comments.

\begin{figure}[h]
\begin{center}
\includegraphics[scale=0.5]{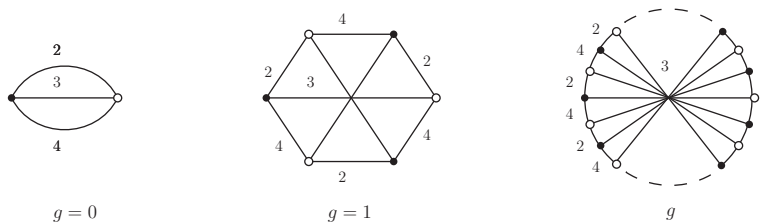} 
\caption{Canonical graphs for orientable surfaces: the sphere ($g=0$), the torus ($g=1$), and the general case of a genus $g$ surface.} \label{bubbles}
\end{center}
\end{figure}

We are ready to build our elementary graphs. For a given genus $g \in \mathbb{N}$ we start from the canonical graph shown in figure (\ref{bubbles}), and add external legs of color $1$ on every node. This gives a set of canonical bubbles with external legs, from which we can in principle construct any topology that is generated by the colored Boulatov model. 

The first question we ask is whether the bound (\ref{bound_core}) in terms of the maximal genus $g_{max}$ is optimal or not. The simplest graph we could think of to saturate this bound consists in one unique canonical graph of genus $g_{max}$ with a pairing of external legs maximizing the amplitude. As a first step, and also because this gives an interesting result that we will use in next section, we first contract $2g$ pairs of external legs, keeping two of them free. More precisely, for each genus $g$ we define the graph $\cC_{g}$ as shown in figure (\ref{bubbles_legs}). 

\begin{figure}[h]
\begin{center}
\includegraphics[scale=0.5]{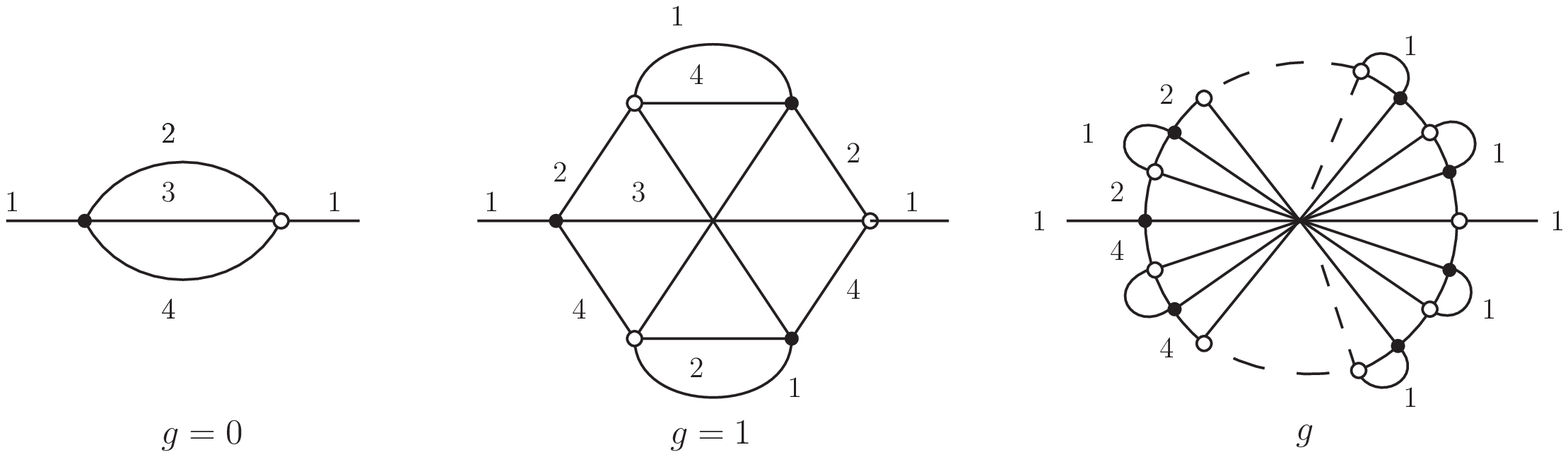} 
\caption{From left to right: $\cC_0$, $\cC_1$ and $\cC_g$.} \label{bubbles_legs}
\end{center}
\end{figure}

The reason why these graphs are useful lies in the following lemma (see figure (\ref{lemma_bubbles})):
\begin{lemma}\label{lemma_cg}
Let $\cG$ be a colored graph, with a subgraph $\cC_{g}$ for some $g \in \mathbb{N}$. Call $\cG^{\cC_{g} \rightarrow \cC_{0}}$ the graph obtained after replacement of $\cC_{g}$ by the graph $\cC_{0}$. Then:
\beq
\cA^{\cG} = (\lambda \overline{\lambda})^{2 g}[\delta^\Lambda(\one)]^{- 2 g} \cA^{\cG^{\cC_{g} \rightarrow \cC_{0}}}
\eeq

\end{lemma}
\begin{proof}
We first remark that a canonical triangulation of a bubble of genus $g$ has only three vertices. Indeed, dual of vertices (of the triangulation) of color $2$ are closed chains of strands alternatively of color $3$ and $4$. It is easy to see that there is only one such closed chain of strands in the canonical graph of genus $g$. So the dual triangulation has only one vertex of color $2$, and similarly for colors $3$ and $4$. This means that in the amplitude of $\cG$, the $\cC_{g}$ subgraph contributes with only three $\delta$-functions associated to its dual vertices, and $2g$ $\delta$-functions associated to pairings of lines of color $1$. Moreover, these pairings are such that the arguments in the $\delta$-functions associated to the vertices simplify, and the contribution of $\cG_g$ to the amplitude of $\cG$ reads:
\beq
(\lambda \overline{\lambda})^{2 g + 1}  \int [\extd H]^{6 g} \left( [\delta^\Lambda(\one)]^{2-2 g - 3} \delta^{\Lambda}(G_{2} {\tilde{G}_{2}}^\inv) \delta^{\Lambda}(G_{3} {\tilde{G}_{3}}^\inv) \delta^{\Lambda}(G_{4} {\tilde{G}_{4}}^\inv) \right) \left( \prod_{i=1}^{2 g} \delta^{\Lambda}(H^{(i)}_2 H^{(i)}_3  H^{(i)}_4 )\right) \,,
\eeq
where the variables $G_{\ell'}$ and $\tilde{G}_{\ell'}$ are that of the two external legs of $\cG_g$. The $H^{(i)}_{\ell'}$ are associated to the $2 g$ remaining lines of color $1$, and can be integrated. We obtain a term:
\beq
(\lambda \overline{\lambda})^{2 g + 1} [\delta^\Lambda(\one)]^{-2 g - 1} \delta^{\Lambda}(G_{2} {\tilde{G}_{2}}^\inv) \delta^{\Lambda}(G_{3} {\tilde{G}_{3}}^\inv) \delta^{\Lambda}(G_{4} {\tilde{G}_{4}}^\inv)\,,
\eeq
which reduces to
\beq
(\lambda \overline{\lambda}) [\delta^\Lambda(\one)]^{- 1} \delta^{\Lambda}(G_{2} {\tilde{G}_{2}}^\inv) \delta^{\Lambda}(G_{3} {\tilde{G}_{3}}^\inv) \delta^{\Lambda}(G_{4} {\tilde{G}_{4}}^\inv)
\eeq
when $g = 0$. These two terms differ by a factor $(\lambda \overline{\lambda})^{2 g} [\delta^\Lambda(\one)]^{-2 g}$, which concludes the proof.
\end{proof}

\begin{figure}[h]
\begin{center}
\includegraphics[scale=0.5]{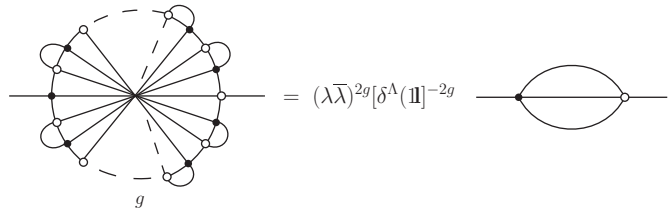} 
\caption{Graphical representation of lemma $1$.} \label{lemma_bubbles}
\end{center}
\end{figure}

This is all we need in order to prove the following theorem:

\begin{theorem}\label{theo1} If a vacuum graph $\cG$ of order $\cN$ contains a bubble of genus $g$, then: 
\beq
\cA^{\cG} \leq (\lambda \overline{\lambda})^{\frac{\cN}{2}}  [\delta^\Lambda(\one)]^{2 - 2 g}\,.
\eeq
Reciprocally, for any $g \in \mathbb{N}$, there exists a graph $\cG$ with at least one bubble of genus $g$, such that:
\beq
\cA^{\cG} = (\lambda \overline{\lambda})^{\frac{\cN}{2}}  [\delta^\Lambda(\one)]^{2 - 2 g}\,,
\eeq
where $\cN$ is the order of $\cG$.
\end{theorem}

\begin{proof}
The first part is a consequence of inequality (\ref{bound_core}). We just have to exhibit a graph that saturates the bound. For $g \in \mathbb{N}$, consider the graph $\cC_{g}$ and join its two external legs. Lemma \ref{lemma_cg} ensures that the amplitude of such a graph is $(\lambda \overline{\lambda})^{2 g}[\delta^\Lambda(\one)]^{- 2 g}$ times the amplitude of the analog graph obtained from $\cC_{0}$. The latter is the sunshine graph, dual to a sphere, and whose amplitude is $(\lambda \overline{\lambda})[\delta^\Lambda(\one)]^{2}$, as can be verified by direct computation. The amplitude we are looking for is therefore $(\lambda \overline{\lambda})^{2 g + 1}  [\delta^\Lambda(\one)]^{2 - 2 g}$, which concludes the proof, by suitable matching of the genus $g$ and the order $\cN$ of the graph. 
\end{proof}

One would now like to generalize this result, for example by constructing a connected graph with singularities of genera $g_{1}, ..., g_{n}$ which scales as $[\delta^\Lambda(\one)]^{{\sum_{i = 1}^{n}} (2 - 2 g_i)}$. Actually, one can show that this is not possible. We will show, instead, that the bound (\ref{bound_core}) can be made sharper, and explain how theorem (\ref{theo1}) generalizes. Moreover, the proof of this refined bound will not rely anymore on dipole contractions and core graphs.

\subsection{Optimal bounds} \label{sec:bubble3}

We start by a computation of the amplitudes of chains of canonical graphs $\cC_{g}$. We call $\cC_{g_1,...,g_{n}}$ the chain of $n$ graphs $(\cC_{g_{1}}, \cdots, \cC_{g_n})$ as represented in figure (\ref{chains}). Chains of $\cC_{0}$ graphs being maximally divergent spheres \cite{scaling3d}, this suggests that the chain $\cC_{g_1,...,g_{n}}$ could be a dominant graph in the class of graphs with singularities $(g_{1}, ..., g_{n})$. So let us first compute these amplitudes.

\begin{figure}[h]
\begin{center}
\includegraphics[scale=0.5]{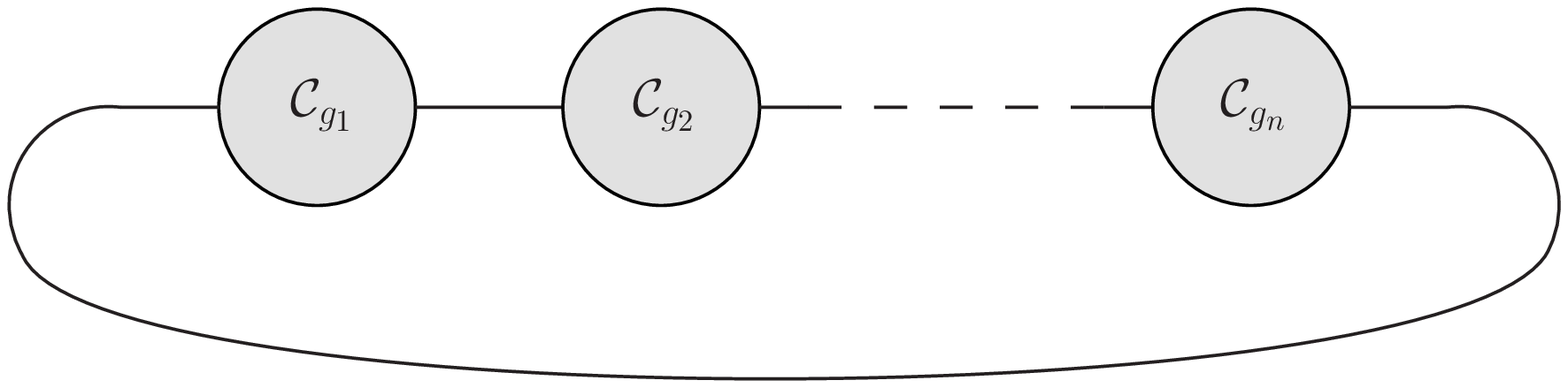} 
\caption{Chain $\cC_{g_1,...,g_{n}}$.} \label{chains}
\end{center}
\end{figure}

\begin{lemma}\label{lemma_chains} Let $n \in \mathbb{N^*}$, and $g_{1}, ..., g_{n} \in \mathbb{N}$. Then:
\beq
\cA^{\cC_{g_1,...,g_{n}}} = (\lambda \overline{\lambda})^{{\sum_{i = 1}^{n}} (2 g_{i}+1)}[\delta^\Lambda(\one)]^{2 - 2 {\sum_{i = 1}^{n}} g_i}\,.
\eeq 
\end{lemma}
\begin{proof}
Lemma \ref{lemma_cg} ensures that $\cC_{g_1,...,g_{n}}$ behaves like $\cC_{0,...,0}$ times $[\delta^\Lambda(\one)]^{-2 {\sum_{i = 1}^{n}} g_i}$. As proved in \cite{scaling3d}, $\cC_{0,...,0}$ is dual to a sphere and maximally divergent. A way to see it here is to remark that with the scaling we chose for the coupling $\lambda$, a $\cC_{0}$ subgraph whose two external strands are not paired behaves like a propagator (straightforward calculation), so that we can replace all subgraphs $C_{0}$ but one by propagators. We are left with the simplest graph corresponding to a sphere, that is again the sunshine graph which behaves like $[\delta^\Lambda(\one)]^{2}$. All in all, we get the right expression for the amplitude of $\cC_{g_1,...,g_{n}}$.
\end{proof}

If we stick to the previous analogy, the maximal amplitudes of graphs with bubbles of genera $(g_{1}, ..., g_{n})$ should be in $[\delta^\Lambda(\one)]^{2 - 2 {\sum_{i = 1}^{n}} {g_i}}$. In the remainder of this section we will show it is indeed the case.

We need to improve the bound (\ref{bound}). Recall that it has been obtained from formula (\ref{final_amplitude}) by first integrating all propagators. Making more precise this step of the procedure, it is possible to integrate more $\delta$-functions in the bubbles before using the bound $\delta^\Lambda(\cdots) \leq \delta^\Lambda(\one)$. This allows to prove the following proposition: 

\begin{proposition}\label{propo_ineg}
Let $\cG$ be a connected vacuum graph of order $\cN$, and $(\ell_1 \; \ell_2 \; \ell_3 \; \ell_4)$ a permutation of the colors. Then:
\beq\label{formula_propo}
\cA^{\cG} \leq (\lambda \overline{\lambda})^{\frac{\cN}{2}}  [\delta^\Lambda(\one)]^{|\cB_{\ell_{2}}| + |\cB_{\ell_{3}}| + {\sum_{b \in \cB_{\ell_1}}} (2 - 2g_b - |V_{b}(l_{2})| - |V_{b}(l_{3})|)}\,,
\eeq
where for any bubble $b \in \cB_{\ell_1}$, and $\ell_i \neq \ell_1$, $V_{b}(\ell_{i})$ denotes the set of vertices of color $\ell_i$ in $b$.
\end{proposition}
\begin{proof}
We start from equation (\ref{final_amplitude}), with $\ell = \ell_1$. We then have to integrate all the $\delta$-functions associated to the triangles of color $\ell_1$. Instead of using arbitrary variables, we integrate all the variables of a given color $\ell_4$. Since there is initially one variable of color $\ell_4$ per propagator of color $\ell_1$, it is possible to integrate all of them. We get an expression of this form:
\beq
\cA^{\cG} = (\lambda \overline{\lambda})^{\frac{\cN }{2}}  \int [\extd G]^{\cN} \left( \prod_{b \in \cB_{\ell_1}} [\delta^\Lambda(\one)]^{2-2 g_b - |V_{b}|} \left(  \prod_{v \in V_{b}(\ell_4)} \delta^\Lambda\left( \cdots \right) \right)  \left( \prod_{v \in V_{b}(\ell_2) \cup V_{b}(\ell_3)} \delta^\Lambda\left( \overrightarrow{\prod_{f \in \triangle^{b}_{v}}} (G_{v}^{f})^{\epsilon^{f}_{v}}\right) \right)  \right)\, .
\eeq
Here dots still denote complicated products of the holonomies and their inverses. Let us now focus on the parts of the integrand involving only variables of colors $\ell_2$ or $\ell_3$. For instance for $\ell_2$ it is:
\beq
\prod_{b \in \cB_{\ell_1}} \left( \prod_{v \in V_{b}(\ell_2)} \delta^\Lambda\left( \overrightarrow{\prod_{f \in \triangle^{b}_{v}}} (G_{v}^{f})^{\epsilon^{f}_{v}}\right) \right)\,.
\eeq
This term is represented by $|\cB_{\ell_{2}}|$ connected graphs of color $\ell_2$. Likewise we have a set of $|\cB_{\ell_{3}}|$ connected graphs of color $\ell_3$. The key point is that all these graphs, which were initially linked through propagators, have now independent variables. Therefore integrating strands in one of them will keep the others unchanged. Of course integrating strands in these graphs will further complicate the remainder of the integrand, that is the terms associated to the vertices of color $\ell_4$. But these are terms we will in the end bound by their value in $\one$, so the precise expression of their arguments is irrelevant.

Let $\cS$ be one of the $|\cB_{\ell_{2}}| + |\cB_{\ell_{3}}|$ such connected graphs. We can choose a maximal tree $\cT$ with root $r$ in $\cS$, and integrate all the strands in this tree. Each of these integrations just deletes a node of the graph, so that in the end only the root remains. The contribution of $\cS$ has been reduced to one unique $\delta$-function, with possibly a very complicated argument, though. Repeating the procedure so as to reduce all the connected graphs of color $\ell_{2}$ and $\ell_{3}$, we get:
\bes
\cA^{\cG} &=& (\lambda \overline{\lambda})^{\frac{\cN }{2}}  \int [\extd G]^{\cN} \left( \prod_{i =1}^{|\cB_{\ell_{2}}| + |\cB_{\ell_{3}}|} \delta^\Lambda\left( \cdots \right) \right) \left( \prod_{b \in \cB_{\ell_1}} [\delta^\Lambda(\one)]^{2-2 g_b - |V_{b}|} \left(  \prod_{v \in V_{b}(\ell_4)} \delta^\Lambda\left( \cdots \right) \right)  \right)\nn \\
&\leq& (\lambda \overline{\lambda})^{\frac{\cN }{2}}  [\delta^\Lambda\left( \one\right)]^{|\cB_{\ell_{2}}| + |\cB_{\ell_{3}}|} \left( \prod_{b \in \cB_{\ell_1}} [\delta^\Lambda(\one)]^{2-2 g_b - |V_{b}| + V_{b}(\ell_4)} \right)\nn 
\ees
But for any $b \in \cB_{\ell_4}$ we have of course $|V_{b}| - |V_{b}(\ell_4)| = |V_{b}(\ell_2)| + |V_{b}(\ell_3)|$, which concludes the proof.
\end{proof}

Starting from this proposition, we just need a bit of work on the combinatorics of a colored graph to arrive at our final result. We do so by proving the following:

\begin{lemma}\label{lemma_comb} 
Let $\cG$ be a connected vacuum graph. Then:
\beq\label{formula_lemma_comb}
\forall \ell \neq \ell'\,, \; |\cB_{\ell'}| + |\cB_{\ell}| - {\sum_{b \in \cB_{\ell}}} |V_{b}(\ell')| \leq 1\,.
\eeq
\end{lemma}
\begin{proof}
Choose two colors $\ell \neq \ell'$. From $\cG$ we construct a connectivity graph $\cC_{\ell, \ell'}(\cG)$, whose elements are the bubbles of color $\ell$ and $\ell'$. Then for any $b' \in \cB_{\ell'}$ and $b \in \cB_{\ell}$ we draw a line between them if and only if $b$ has a vertex dual to $b'$ in its triangulation. We call $L$ the number of lines of $\cC_{\ell, \ell'}(\cG)$, and $N$ its number of elements. 
Now remark that the fact that $\cG$ is connected implies that $\cC_{\ell, \ell'}(\cG)$ is also connected. In fact, the bubbles of color $\ell'$ are all connected in $\cG$ by lines of color $\ell'$. But these lines are themselves part of bubbles of color $\ell$, which means that two bubbles of color $\ell'$ are connected if and only if their dual vertices appear in a same bubble of color $\ell$, that is if and only if they are connected to a same element in the graph $\cC_{\ell, \ell'}(\cG)$. So $\cC_{\ell, \ell'}(\cG)$ is connected. A maximal tree in this graph has $N -1$ lines, which implies the simple inequality: $N- 1 \leq L$.
To conclude, first notice that by construction $N$ is equal to $|\cB_{\ell'}| + |\cB_{\ell}|$. Still by construction, for any $b \in \cB_{\ell}$, $|V_{b}(\ell')|$ has to be greater than the number of lines ending on $b$ in $\cC_{\ell, \ell'}(\cG)$. Therefore:
\beq
|\cB_{\ell'}| + |\cB_{\ell}| - {\sum_{b \in \cB_{\ell}}} |V_{b}(\ell')| \leq N - L \leq 1\,.
\eeq
\end{proof}

The next, concluding theorem follows easily from the previous results. For what concerns bounds on quantum amplitudes, it is the main result of this paper.

\begin{theorem}\label{theo2} Let $\cG$ be a connected vacuum graph of order $\cN$. Then for any color $\ell$: 
\beq\label{best_ineq}
\cA^{\cG} \leq (\lambda \overline{\lambda})^{\frac{\cN}{2}}  [\delta^\Lambda(\one)]^{2 - 2 {\sum_{b \in \cB_{\ell}}} g_b}\,.
\eeq
Reciprocally, for any integers $(g_{1},...,g_{n})$, there exists a graph $\cG$ whose bubbles of color $\ell$ have genera $(g_{1},...,g_{n})$, and such that:
\beq
\cA^{\cG} = (\lambda \overline{\lambda})^{\frac{\cN}{2}}  [\delta^\Lambda(\one)]^{2 - 2 {\sum_{i=1}^{n}} g_i}\,,
\eeq
where $\cN$ is the order of $\cG$.
\end{theorem}
\begin{proof} The first part of the theorem is a consequence of proposition \ref{propo_ineg} and lemma \ref{lemma_comb}, easily proven as follows: apply formula (\ref{formula_propo}) with $\ell_1 = \ell$, and any other colors $\ell_2$ and $\ell_3$; then bound the exponent of $\delta^\Lambda(\one)$ using two times the inequality (\ref{formula_lemma_comb}). As for the second part, this is exactly the content of lemma \ref{lemma_chains}.
\end{proof}

As already mentioned in the introduction of this section, this result does not rely on $1$-dipole contractions. Instead the stranded structure of the graphs in vertex variables allowed us to perform many integrals before using any inequality. The bound so derived is optimal, which means that this procedure fully captures the properties of the bubbles (of a given color), in the sense that the result could not be improved without taking the overall topology of the pseudo-manifold into account. Note finally that in order to compare our results to previous ones in the literature, the rescaling of the coupling (\ref{rescaling_coupling}) has to be taken into account. With this in mind, formula (\ref{best_ineq}) is indeed consistent with the existing literature \cite{lrd, scaling3d, vincentcolored, linearizedgft, ValentinMatteo3}.

\section{Conclusions and outlook}
We have proven new scaling bounds for the amplitudes of the (colored, bosonic) Boulatov group field theory for 3d quantum gravity, dependent on the bubble structure of  the associated simplicial complexes. More precisely the new bounds depend on the number and type of the point singularities of the same complex (higher-dimensional singularities have been shown to be absent in these models), and thus measure in a sense its degree of \lq manifold-ness\rq . Moreover, we have shown these bounds to be optimal. Accordingly, manifold configurations dominate, for large values of the cut-off, over pseudo-manifold configurations. 

These results deepen our understanding of the GFT perturbative expansion, and the associated sum over simplicial complexes, and confirm that  group field theories have the potential, indeed, to realize in higher dimension the success story of matrix models in two dimensions, with the emergence of a smooth spacetime and gravity from a more fundamental, pre-geometric, quantum system. 

This is thanks to the additional data turning simpler tensor models into proper field theories, and to the insights provided by loop quantum gravity and simplicial quantum gravity concerning the nature of these additional data, as encoding the properties of a quantum spacetime. In fact, in this paper we have taken full advantage of the recently developed non-commutative metric formulation of GFTs \cite{aristidedaniele}, in turn motivated by results in loop quantum gravity and spin foam models. In particular we adopted and applied to our task a re-writing of the same Boulatov GFT model in terms of vertex variables, which was suggested by the identification of simplicial diffeomorphism symmetry at the GFT level \cite{diffeos}.

Two immediate questions arise, concerning our results. As we commented above, there are several indications (see the discussion in \cite{diffeos}) that a non-trivial braiding could be, if not needed, certainly natural in a GFT context. A priori this could affect the structure and evaluation of the GFT amplitudes, and thus their scaling. It would be interesting to check, in particular, how it could affect the relative weight of manifolds and pseudo-manifolds. Also interesting would be to investigate whether our scaling bounds related to manifold-ness could be derived from the very nice power counting results based on twisted cohomology \cite{ValentinMatteo3}, under some assumption on the overall topology of the diagrams, and conversely, whether our bounds, together with these power counting results, allow to understand in more detail the scaling of the GFT amplitudes with the topology of the diagrams, and thus improve the existing results on the large cut-off expansion \cite{RazvanFullN}.

More generally, we expect the vertex reformulation of GFTs to lend itself to more applications, and to be useful for elucidating further the geometry of GFT models as well as their effective dynamics, and, possibly, to be the basis for a \lq first principles{\rq} definition of GFT models. 

The main open question is whether and how this vertex formulation of GFTs, and our results, generalize to higher dimensions, for topological models as well as, more importantly, for 4d gravity models; progress along these lines will have to proceed alongside progress in our understanding of GFT symmetries for 4d gravity models. We believe our results, and previous ones concerning these matters, indicate a promising direction.

We have therefore reasons to hope that the rapidly accumulating wealth of results concerning both the quantum geometry behind GFT models and the GFT perturbative expansion, together with the development of appropriate non-perturbative tools, will trigger much further progress in this area, in particular concerning their continuum limit, phase structure and effective (quantum) gravitational dynamics, toward a complete understanding of quantum spacetime.

\section*{Acknowledgements}
We thank A. Baratin, F. Caravelli,  R. Gurau, M. Raasakka and V. Rivasseau for useful comments and discussions.
DO gratefully acknowledges financial support from the A. von Humboldt Stiftung through a Sofja Kovalevskaja Prize.


\end{document}